\documentclass[11pt]{article}

\usepackage{amsmath,amssymb,amsthm,mathtools}
\usepackage{bm}
\usepackage{booktabs}
\usepackage{enumitem}
\usepackage[margin=1.05in]{geometry}
\usepackage{microtype}
\usepackage{natbib}
\usepackage{setspace}
\usepackage{hyperref}

\hypersetup{
  colorlinks=true,
  linkcolor=blue,
  citecolor=blue,
  urlcolor=blue
}

\newtheorem{assumption}{Assumption}
\newtheorem{theorem}{Theorem}
\newtheorem{proposition}{Proposition}
\newtheorem{lemma}{Lemma}
\newtheorem{corollary}{Corollary}
\newtheorem{example}{Example}
\theoremstyle{remark}
\newtheorem{remark}{Remark}

\newcommand{\E}{\mathbb{E}}
\newcommand{\Pp}{\mathbb{P}}
\newcommand{\Var}{\operatorname{Var}}
\newcommand{\Cov}{\operatorname{Cov}}
\newcommand{\R}{\mathbb{R}}
\newcommand{\1}{\mathbf{1}}

\newcommand{\argmin}{\operatorname*{arg\,min}}

\newcommand{\toD}{\mathrel{\overset{d}{\longrightarrow}}}
\newcommand{\toP}{\mathrel{\overset{p}{\longrightarrow}}}

\title{Stationary Errors and Quantile Regression in Short Panels\footnote{We thank seminar participants at the University of Toronto and conference participants at  the 2025 Advances in Econometrics conference in Aarhus, Denmark,  2025 SEA meetings in Tampa, Florida, the 2026 AFES in Cairo, Egypt, the 2026 James G. MacKinnon 75th Birthday Conference in Aarhus, Denmark, the 2026 ISNPS in Thessaloniki, Greece,  for helpful comments. We are also grateful to Shengbin Wei for very helpful research assistance. }}
\author{Shakeeb Khan \\ Boston College \and Elie Tamer \\ Harvard University}
\date{August 2026}

\begin{document}
\maketitle

\begin{abstract}
This paper studies a linear panel model with an unrestricted individual effect and a time-stationary idiosyncratic disturbance. We first show that stationarity is a strong restriction in a quantile model. In a linear conditional quantile specification with quantile-dependent slopes, equality of the conditional residual distributions across periods generically forces the slope coefficient to be constant over the quantile index. Thus, a stationary-error model identifies a common location coefficient rather than a collection of quantile-specific slope effects. We then develop a fixed-$T$ estimator of this common coefficient. For each period, we run a cross-sectional quantile regression of the outcome on the full history of regressors. Stationarity makes the quantile projection of the composite individual effect and disturbance common across the period-specific regressions. Differences between diagonal and off-diagonal blocks of the resulting projection coefficients therefore identify the common slope whenever $T\geq 2$. We combine all such restrictions by a two-step minimum-distance estimator. The estimator is $\sqrt{n}$-consistent and asymptotically normal with fixed $T$, permits unrestricted dependence across periods within an individual, and does not estimate the individual effects. We provide a consistent analytic covariance estimator, a cluster bootstrap, and an overidentification test of the projection restrictions implied by stationarity. Extensive Monte Carlo experiments show adequate performance under various designs.
\end{abstract}

\noindent\textbf{Keywords:} panel data; quantile regression; stationarity; fixed effects; minimum distance; short panels.

\noindent\textbf{JEL classification:} C13, C14, C23.
\newpage 
\section{Introduction}
\label{sec:introduction}

Quantile regression is attractive in panel data because it permits observed covariates to have different associations with different parts of the conditional outcome distribution. The individual effect that motivates the use of panel data, however, creates a familiar difficulty. Unlike a conditional mean restriction, a conditional quantile restriction is not preserved by first differencing. Direct fixed-effects quantile regression therefore introduces one nuisance parameter per individual and is typically analyzed under asymptotic sequences in which the number of time periods grows. This is the setting of the penalized estimator in \citet{koenker2004}, \citet{lamarche2010}, the two-step transformation in \citet{canay2011} and \citet{besstregolovan2019}, and the large-$n$, large-$T$ theory developed by \citet{katoetal2012}, \citet{galvaokato2016}, and \citet{galvaoetal2020}.

This paper considers a different model and a different asymptotic experiment. For individuals $i=1,\ldots,n$ and periods $t=1,\ldots,T$, let
\begin{equation}
Y_{it}=X_{it}'\beta_0+A_i+V_{it},
\label{eq:location-model}
\end{equation}
and impose conditional time stationarity of the idiosyncratic disturbance,
\begin{equation}
V_{it}\mid (X_i,A_i)\ \overset{d}{=}\ V_{is}\mid (X_i,A_i),
\qquad s,t\in\{1,\ldots,T\},
\label{eq:stationarity}
\end{equation}
where $X_i=(X_{i1}',\ldots,X_{iT}')'$ is the full regressor history. The number of periods is fixed while $n$ tends to infinity. The distribution of $A_i$ is unrestricted and may be arbitrarily related to $X_i$. The disturbances may be serially dependent, and their common conditional distribution may depend on both $X_i$ and $A_i$. In particular, \eqref{eq:stationarity} is not a mean-independence or a median-independence condition. It is a cross-period equality restriction on conditional distributions.

We begin with the implication of this restriction for quantile heterogeneity. Consider the familiar linear conditional quantile representation
\[
Q_\tau(Y_{it}\mid X_i,A_i)
 =c(\tau)+X_{it}'\beta(\tau)+\rho(\tau)A_i.
\]
For a fixed $\tau$, define the corresponding quantile residual by subtracting the right-hand side. If the conditional distribution of this residual is required to be the same in two periods, then equality of all of its conditional quantiles implies
\[
(X_{it}-X_{is})'\{\beta(u)-\beta(\tau)\}=0
\qquad\text{for every }u\in(0,1).
\]
Under a standard rank condition on within-individual regressor changes, this equality forces $\beta(u)=\beta(\tau)$ for every $u$. Stationarity at even one quantile index therefore rules out quantile-varying slopes. Under a stronger no-hyperplane-mass condition, if $\beta(\cdot)$ is not constant, equality of the two residual distributions can occur only on a null set of regressor histories. This result is useful both substantively and methodologically. It clarifies that a model combining quantile-varying slopes with stationary quantile-specific residual distributions is generically internally inconsistent. It also clarifies the interpretation of $\beta_0$ in \eqref{eq:location-model}: it is a common location coefficient, not a quantile-specific effect.

The second contribution is a fixed-$T$ identification argument for $\beta_0$. Let $X_i=(X_{i1}',\ldots,X_{iT}')'$ and, for a chosen quantile index $q$, run a separate population quantile projection of $Y_{it}$ on $X_i$ for each period. Under \eqref{eq:location-model}--\eqref{eq:stationarity}, the projection coefficient in period $t$ has the form
\[
\pi_{0t}(q)=\delta_0(q)+E_t\beta_0,
\]
where $E_t=e_t\otimes I$ places $\beta_0$ in the block corresponding to $X_{it}$ and $\delta_0(q)$ is common across periods. The common nuisance coefficient is the quantile projection of $A_i+V_{it}$ on the full regressor history. It need not equal a conditional quantile coefficient and it is allowed to vary with $q$. For every ordered pair $s\neq t$, the contrast between the coefficient on $X_{is}$ in the period-$s$ projection and the coefficient on the same regressor block in the period-$t$ projection equals $\beta_0$. Thus $T\geq 2$ is sufficient for point identification.

The third contribution is estimation and inference. The first step consists of $T$ ordinary cross-sectional quantile regressions, each using the $n$ individuals. The second step pools all $T(T-1)$ slope contrasts by minimum distance. This construction avoids estimating the $n$ individual effects and therefore has standard $\sqrt{n}$ behavior with fixed $T$. We derive the joint influence function of the first-step quantile projections under possible misspecification of the conditional quantile function. This point matters because the first-step coefficients are best linear quantile projections; they need not describe the true conditional quantile given $X_i$. The appropriate covariance matrix is consequently a robust quantile-regression sandwich with cross-period score covariances. The middle component of this sandwich is formed from products of quantile scores, not indicators that a fitted residual is exactly zero.

The minimum-distance formulation also produces specification restrictions. There are $T(T-1)$ vector-valued estimates of the same $\beta_0$. The optimally weighted distance between these estimates and their common fitted value has an asymptotic chi-square distribution. In particular, the full system is overidentified already when $T=2$: the two diagonal-minus-off-diagonal contrasts must agree. For $T>2$, stationarity additionally requires equality of the off-diagonal projection coefficients within each regressor block.

Our analysis is related to several parts of the panel quantile and nonlinear panel literature. \citet{rosen2012} studies set identification in short panels when conditional quantile restrictions are combined with weak restrictions on cross-period dependence. \citet{chernozhukovetal2013} use time homogeneity to identify or bound average and quantile effects in nonseparable panel models, while \citet{chernozhukovetal2015} identify structural quantile derivatives for stayers under time homogeneity. \citet{arellanobonhomme2016} develop simulation-based quantile methods for nonlinear short panels, and \citet{grahametal2018} study quantile correlated random coefficients. More directly related to our maintained restriction, \citet{chenwang2018} exploit conditional stationarity in short nonlinear panels and construct minimum-distance estimators by matching conditional outcome distributions. \citet{botosarumuris2025} use conditional time stationarity to sharpen identification of counterfactual distributions in a broader class of nonlinear panel models. The present paper instead uses stationarity in the additive location model to obtain linear restrictions on period-specific cross-sectional quantile projections. These restrictions deliver a closed-form second step and, separately, reveal the incompatibility between stationary quantile-specific residual distributions and quantile-varying slopes. The use of stationarity is also in the spirit of \citet{manski1987}, although the observable implications and the estimator are different.

The projection-and-minimum-distance structure has a classical antecedent in the multivariate linear-predictor analysis of \citet{chamberlain1982}. It is also related to the short-panel identification perspective in \citet{komarovaseverinitamer}. There is also a literature on minimum-distance panel quantile estimation. \citet{galvaowang2015} combine unit-specific time-series quantile regressions and establish results under large-$T$ asymptotics. The recent procedure of \citet{mellypons2026} likewise begins with within-unit quantile regressions and lets the number of observations per unit grow. In contrast, our first-step regressions are period-specific cross-sectional regressions. Each is estimated from $n$ independent individuals, and $T$ remains fixed. We are not aware of an existing procedure that uses equality of these period-specific cross-sectional quantile projections to obtain fixed-$T$ point estimation and root-$n$ inference for the common location slope. The quantile-projection interpretation follows the analysis of potentially misspecified quantile regression in \citet{angristetal2006}.

The remainder of the paper is organized as follows. Section \ref{sec:incompatibility} establishes the incompatibility between residual stationarity and quantile-varying slopes. Section \ref{sec:identification} derives identification from period-specific quantile projections. Section \ref{sec:estimation} introduces the two-step minimum-distance estimator. Section \ref{sec:asymptotics} develops limit theory, feasible variance estimation, bootstrap inference, and specification testing. Section \ref{sec:multiple-q} shows how to combine several quantile projections. Section \ref{mc:sec} explores finite sample properties via a  Monte Carlo analysis, and Section \ref{sec:conclusion} concludes. Proofs are collected in the Appendix.

\section{Stationarity and quantile heterogeneity}
\label{sec:incompatibility}

This section isolates the restriction that stationarity places on a conventional linear panel quantile specification. It is logically prior to the estimation argument: before using stationarity, one must determine which quantile objects can vary under that assumption.

\subsection{A linear conditional quantile specification}

Let $\mathcal{T}\subset(0,1)$ be an interval. Suppose that, for every $\tau\in\mathcal{T}$ and every period $t$,
\begin{equation}
Q_\tau(Y_{it}\mid X_i,A_i)
 =c(\tau)+X_{it}'\beta(\tau)+\rho(\tau)A_i.
\label{eq:linear-quantile-model}
\end{equation}
Here $Q_\tau$ denotes the generalized inverse conditional quantile, and $c(\tau)$ is an unrestricted quantile-specific intercept common across periods. We assume that \eqref{eq:linear-quantile-model} is a valid version of the conditional quantile process on a common probability-one set. Define the $\tau$-specific residual
\begin{equation}
\varepsilon_{it}(\tau)
 =Y_{it}-c(\tau)-X_{it}'\beta(\tau)-\rho(\tau)A_i.
\label{eq:tau-residual}
\end{equation}
By construction, $Q_\tau(\varepsilon_{it}(\tau)\mid X_i,A_i)=0$. The relevant stationarity requirement is stronger than this zero-quantile normalization.

\begin{assumption}[Residual stationarity at a quantile]
\label{ass:residual-stationarity}
For a pair $t\neq s$ and an index $\tau\in\mathcal{T}$,
\begin{equation}
\varepsilon_{it}(\tau)\mid(X_i,A_i)
 \overset{d}{=}
\varepsilon_{is}(\tau)\mid(X_i,A_i)
\qquad\text{a.s.}
\label{eq:residual-stationarity}
\end{equation}
\end{assumption}

For later use, let $\mathcal S_{ts}(\tau)$ denote the set of conditioning values $(x,a)$ at which the two conditional distributions in \eqref{eq:residual-stationarity} coincide. Assumption \ref{ass:residual-stationarity} is equivalently
$\Pp\{(X_i,A_i)\in\mathcal S_{ts}(\tau)\}=1$.

Let $\Delta X_{its}=X_{it}-X_{is}$. The following support condition is sufficient for the slope difference in any direction to be detected by within-individual changes.

\begin{assumption}[Within-individual rank]
\label{ass:within-rank}
For at least one pair $t\neq s$,
\begin{equation}
\E[\Delta X_{its}\Delta X_{its}']
\quad\text{is positive definite.}
\label{eq:within-rank}
\end{equation}
\end{assumption}

\begin{theorem}[Stationarity rules out quantile-varying slopes]
\label{thm:incompatibility}
Suppose \eqref{eq:linear-quantile-model} holds. Fix $t\neq s$ and $\tau,u\in\mathcal{T}$.

\begin{enumerate}[label=(\roman*)]
\item On every conditioning value $(X_i,A_i)\in\mathcal S_{ts}(\tau)$,
\begin{equation}
\Delta X_{its}'\{\beta(u)-\beta(\tau)\}=0.
\label{eq:hyperplane-implication}
\end{equation}

\item If Assumption \ref{ass:residual-stationarity} holds at one index $\tau\in\mathcal{T}$ and Assumption \ref{ass:within-rank} holds for the same pair, then
\begin{equation}
\beta(u)=\beta(\tau)
\qquad\text{for every }u\in\mathcal{T}.
\label{eq:constant-beta}
\end{equation}
Consequently, residual stationarity at a single quantile forces the slope function $\beta(\cdot)$  to be constant on $\mathcal{T}$.
\end{enumerate}
\end{theorem}

The rank condition in Theorem \ref{thm:incompatibility} is weaker than requiring that every hyperplane receive zero probability. The stronger condition gives a sharper generic-failure statement.

\begin{corollary}[Generic failure under no hyperplane mass]
\label{cor:generic-failure}
Suppose that for a pair $t\neq s$,
\begin{equation}
\Pp(\Delta X_{its}'\gamma=0)=0
\qquad\text{for every }\gamma\neq 0.
\label{eq:no-hyperplane}
\end{equation}
If $\beta(\cdot)$ is not constant on $\mathcal{T}$, then, for every fixed $\tau\in\mathcal{T}$,
\begin{equation}
\Pp\{(X_i,A_i)\in\mathcal S_{ts}(\tau)\}=0.
\label{eq:generic-nonstationarity}
\end{equation}
\end{corollary}

\begin{example}[Scalar regressor and uniform conditional distributions]
\label{ex:uniform}
Let $T=2$, let $X_{it}$ be scalar and positive, set $\rho(\tau)=1$, and suppose
\[
\beta(\tau)=a+b\tau,
\qquad b>0.
\]
For simplicity take $A_i=0$ and let the conditional quantile function be
\[
Q_\tau(Y_{it}\mid X_i)=X_{it}(a+b\tau).
\]
For a fixed $\tau$, the residual in \eqref{eq:tau-residual} has conditional quantile function
\[
Q_u(\varepsilon_{it}(\tau)\mid X_i)
 =X_{it}b(u-\tau).
\]
Thus $\varepsilon_{it}(\tau)\mid X_i$ is uniform on
\[
[-X_{it}b\tau,\;X_{it}b(1-\tau)].
\]
Whenever $X_{i1}\neq X_{i2}$, the two residual supports differ. The residual distributions therefore cannot be stationary across the two periods.
\end{example}

\begin{remark}[The role of the individual effect]
The terms $c(u)-c(\tau)$ and $[\rho(u)-\rho(\tau)]A_i$ cancel when the conditional quantiles of the two residuals are compared. Allowing either a quantile-specific common intercept or the scale of the individual effect to vary with the quantile index does not relax the restriction on $\beta(\cdot)$. The result applies both to the location-shift model $\rho(\tau)\equiv 1$ and to the more general specification in \eqref{eq:linear-quantile-model}.
\end{remark}

\begin{remark}[Intercepts and time-invariant regressors]
The rank condition concerns directions with within-individual regressor variation. It cannot hold for a deterministic intercept or for a regressor that is constant over time. Stationarity therefore does not force the common quantile intercept $c(\tau)$ to be constant; that term may trace the quantiles of a stationary location disturbance. More generally, Theorem \ref{thm:incompatibility} forces constancy only for linear combinations of $\beta(\tau)$ that are identified by the support of $X_{it}-X_{is}$. In the estimation sections below, deterministic intercepts and time-invariant regressors are absorbed by the common nuisance projection and are excluded from the structural slope vector.
\end{remark}

\begin{remark}[What the theorem does and does not say]
Theorem \ref{thm:incompatibility} does not rule out panel quantile models with quantile-varying slopes. It rules out combining such slopes with equality of the full conditional distributions of the quantile-specific residuals. Standard fixed-effects quantile regression imposes a zero conditional quantile restriction and need not impose \eqref{eq:residual-stationarity}. Conversely, a stationary-error model such as \eqref{eq:location-model}--\eqref{eq:stationarity} permits the common conditional distribution of the disturbance to be heterogeneous in $(X_i,A_i)$, but it does not permit a separate slope coefficient at each quantile.
\end{remark}

\begin{remark}[Differencing]
Quantile restrictions are not generally preserved by differencing. If one attempts to remove $A_i$ by differencing and then invokes symmetry or another distributional property derived from equality of the two residual distributions, Theorem \ref{thm:incompatibility} shows that the resulting argument is incompatible with nonconstant $\beta(\tau)$ under ordinary regressor variation. This is one reason to distinguish carefully between a quantile restriction and a stationarity restriction on an entire distribution.
\end{remark}

\section{Identification from stationary quantile projections}
\label{sec:identification}

We now return to the stationary location model \eqref{eq:location-model}--\eqref{eq:stationarity}. The parameter of interest is the common slope $\beta_0\in\R^p$. The individual effect $A_i$ is unrestricted and is not estimated.

\subsection{Notation and assumptions}

Let
\begin{equation}
X_i=(X_{i1}',\ldots,X_{iT}')'\in\R^{pT},
\label{eq:stacked-regressor}
\end{equation}
and define the selector
\begin{equation}
E_t=e_t\otimes I_p\in\R^{pT\times p},
\qquad X_{it}=E_t'X_i.
\label{eq:selector}
\end{equation}
The regressors should be represented without duplicated deterministic columns. In particular, a common intercept or a time-invariant regressor that is absorbed by $A_i$ should not be repeated in every block of $X_i$.

For $q\in(0,1)$, define the check function and its score by
\begin{equation}
\rho_q(u)=u\{q-\1(u<0)\},
\qquad
\psi_q(u)=q-\1(u\leq 0).
\label{eq:check-score}
\end{equation}
The median case in the original construction corresponds to $q=1/2$ and ordinary LAD. We state the results for a generic $q$ because every quantile projection identifies the same $\beta_0$ under stationarity.

\begin{assumption}[Sampling and moments]
\label{ass:sampling}
The vectors $(Y_{i1},\ldots,Y_{iT},X_i,A_i,V_{i1},\ldots,V_{iT})$ are independent and identically distributed across $i$. The integers $T$ and $p$ are fixed. For the chosen $q$,
\[
\E\|X_i\|<\infty,
\qquad
\E|Y_{it}|<\infty,
\quad t=1,\ldots,T.
\]
\end{assumption}

\begin{assumption}[Conditional stationarity]
\label{ass:stationarity}
Equation \eqref{eq:stationarity} holds for every pair $s,t$.
\end{assumption}

Define the composite disturbance $R_{it}=A_i+V_{it}$. Conditional stationarity of $V_{it}$ given $(X_i,A_i)$ implies stationarity of $R_{it}$ given $X_i$ after integrating out $A_i$.

\begin{lemma}[Stationarity of the composite disturbance]
\label{lem:composite-stationarity}
Under Assumption \ref{ass:stationarity},
\begin{equation}
R_{it}\mid X_i\ \overset{d}{=}\ R_{is}\mid X_i
\qquad\text{for all }s,t.
\label{eq:composite-stationarity}
\end{equation}
\end{lemma}

For each $t$, define the population linear quantile projection
\begin{equation}
\pi_{0t}(q)
 =\argmin_{\pi\in\R^{pT}}
\E\big[\rho_q(Y_{it}-X_i'\pi)\big].
\label{eq:population-first-stage}
\end{equation}
Also define
\begin{equation}
\delta_0(q)
 =\argmin_{\delta\in\R^{pT}}
\E\big[\rho_q(R_{i1}-X_i'\delta)\big].
\label{eq:delta-projection}
\end{equation}
By Lemma \ref{lem:composite-stationarity}, the objective defining $\delta_0(q)$ is the same in every period.

\begin{assumption}[Unique quantile projections]
\label{ass:unique-projection}
The minimizer in \eqref{eq:delta-projection} is unique. The matrix $\E[X_iX_i']$ is finite and positive definite.
\end{assumption}

The next proposition is the central population restriction.

\begin{proposition}[Common-nuisance representation]
\label{prop:projection-representation}
Under Assumptions \ref{ass:sampling}--\ref{ass:unique-projection}, for every $t$,
\begin{equation}
\pi_{0t}(q)=\delta_0(q)+E_t\beta_0.
\label{eq:projection-representation}
\end{equation}
Equivalently, if $\pi_{0t}^{(s)}(q)=E_s'\pi_{0t}(q)$ denotes the coefficient on $X_{is}$ in the period-$t$ projection, then
\begin{equation}
\pi_{0t}^{(s)}(q)
 =\delta_0^{(s)}(q)+\1(t=s)\beta_0.
\label{eq:block-representation}
\end{equation}
\end{proposition}

The representation immediately yields a collection of observable slope contrasts. For every ordered pair $(s,t)$ with $s\neq t$, define
\begin{equation}
b_{0,st}(q)
 =E_s'\{\pi_{0s}(q)-\pi_{0t}(q)\}.
\label{eq:population-contrast}
\end{equation}
Then
\begin{equation}
b_{0,st}(q)=\beta_0,
\qquad s\neq t.
\label{eq:all-contrasts-beta}
\end{equation}
There are $T(T-1)$ such $p$-vectors.

\begin{theorem}[Fixed-$T$ point identification]
\label{thm:identification}
Suppose Assumptions \ref{ass:sampling}--\ref{ass:unique-projection} hold.

\begin{enumerate}[label=(\roman*)]
\item If $T\geq 2$, $\beta_0$ is point identified by any one of the contrasts in \eqref{eq:population-contrast}. The remaining contrasts are overidentifying restrictions.

\item If $T=1$, the projection identifies only $\delta_0(q)+\beta_0$; the two components are not separately identified by this argument.
\end{enumerate}
\end{theorem}

\begin{remark}[The $T=2$ case]
When $T=2$,
\[
\pi_{01}(q)
 =\begin{pmatrix}\delta_0^{(1)}(q)+\beta_0\\ \delta_0^{(2)}(q)\end{pmatrix},
\qquad
\pi_{02}(q)
 =\begin{pmatrix}\delta_0^{(1)}(q)\\ \delta_0^{(2)}(q)+\beta_0\end{pmatrix}.
\]
Hence
\[
\pi_{01}^{(1)}(q)-\pi_{02}^{(1)}(q)=\beta_0,
\qquad
\pi_{02}^{(2)}(q)-\pi_{01}^{(2)}(q)=\beta_0.
\]
One contrast is sufficient for identification, but the full model is not exactly identified. It imposes the additional restriction that these two estimates agree. This gives $p$ overidentifying restrictions when $T=2$.
\end{remark}

\begin{remark}[A weaker identifying restriction]
Full conditional stationarity is sufficient but not necessary for Proposition \ref{prop:projection-representation}. The estimator requires only that the population $q$-quantile projection of $R_{it}$ on $X_i$ be the same for every $t$. We refer to this weaker condition as \emph{quantile-projection stationarity}. The stronger distributional condition in Assumption \ref{ass:stationarity} is attractive because it is invariant to the chosen quantile index and supplies additional testable implications.
\end{remark}

\begin{remark}[No conditional quantile specification is required]
The linear projection in \eqref{eq:population-first-stage} need not equal the true conditional quantile of $Y_{it}$ given $X_i$. The term $\delta_0(q)'X_i$ is the best linear predictor of $A_i+V_{it}$ under check loss; it is not, in general, the projection of the conditional quantile function under a squared-error metric. This distinction is important for robust covariance estimation below.
\end{remark}

\begin{remark}[Mean-projection analogue and the role of quantiles]
If second moments exist, the same population algebra applies to period-specific least-squares projections. In that case stationarity implies a common linear projection of $A_i+V_{it}$ on $X_i$, and diagonal-minus-off-diagonal OLS coefficients also identify $\beta_0$. The quantile-projection construction is therefore not a device for estimating a heterogeneous slope $\beta(\tau)$; Theorem \ref{thm:incompatibility} explains why such an interpretation is unavailable under residual stationarity. Its advantages are that it remains well defined without second moments of the outcome, permits robust location summaries at any $q$, and allows the stationarity restrictions to be compared across several quantile projections. The minimum-distance formulation is the check-loss analogue of imposing cross-equation restrictions on multivariate linear predictors as in \citet{chamberlain1982}.
\end{remark}

\section{Two-step minimum-distance estimation}
\label{sec:estimation}

\subsection{First-step quantile regressions}

For each $t=1,\ldots,T$, estimate \eqref{eq:population-first-stage} by
\begin{equation}
\widehat\pi_t(q)
 =\argmin_{\pi\in\R^{pT}}
\frac{1}{n}\sum_{i=1}^n
\rho_q(Y_{it}-X_i'\pi).
\label{eq:sample-first-stage}
\end{equation}
Stack the first-step coefficients as
\begin{equation}
\widehat\Pi(q)
 =\big(\widehat\pi_1(q)',\ldots,\widehat\pi_T(q)'\big)'
 \in\R^{pT^2}.
\label{eq:stacked-pi-hat}
\end{equation}
For each ordered pair $s\neq t$, form
\begin{equation}
\widehat b_{st}(q)
 =E_s'\{\widehat\pi_s(q)-\widehat\pi_t(q)\}.
\label{eq:sample-contrast}
\end{equation}
Let $r=T(T-1)$ and stack these contrasts in a fixed order:
\begin{equation}
\widehat b(q)
 =\big(\widehat b_{st}(q)':s\neq t\big)'
 \in\R^{pr}.
\label{eq:stacked-b}
\end{equation}
Define
\begin{equation}
D=\1_r\otimes I_p\in\R^{pr\times p}.
\label{eq:D-matrix}
\end{equation}
The population restriction is
\begin{equation}
b_0(q)=D\beta_0.
\label{eq:contrast-system}
\end{equation}

It is useful to write the contrast operation as a matrix. Let $C\in\R^{pr\times pT^2}$ have a block row
\begin{equation}
C_{st}=(e_s-e_t)'\otimes E_s',
\qquad s\neq t.
\label{eq:C-block}
\end{equation}
Then
\begin{equation}
\widehat b(q)=C\widehat\Pi(q),
\qquad
b_0(q)=C\Pi_0(q).
\label{eq:C-representation}
\end{equation}
The matrix $C$ has full row rank $pr$.

\subsection{Equal-weight and efficient minimum distance}

A transparent preliminary estimator averages all diagonal-minus-off-diagonal contrasts:
\begin{equation}
\widehat\beta_{\mathrm{EW}}(q)
 =(D'D)^{-1}D'\widehat b(q)
 =\frac{1}{T(T-1)}\sum_{s=1}^T\sum_{t\neq s}
 E_s'\{\widehat\pi_s(q)-\widehat\pi_t(q)\}.
\label{eq:equal-weight-estimator}
\end{equation}
This estimator uses every period as both a diagonal equation and a comparison equation. A baseline-period estimator that uses only a subset of the contrasts is also consistent, but it generally discards information.

Let
\begin{equation}
\sqrt{n}\{\widehat b(q)-b_0(q)\}
 \toD N(0,\Omega_q),
\label{eq:b-limit-definition}
\end{equation}
where $\Omega_q$ is positive definite. Given a positive-definite weight $A_n\toP A$, define
\begin{equation}
\widehat\beta_A(q)
 =\argmin_{\beta\in\R^p}
 \{\widehat b(q)-D\beta\}'A_n
 \{\widehat b(q)-D\beta\}.
\label{eq:generic-md}
\end{equation}
The closed-form solution is
\begin{equation}
\widehat\beta_A(q)
 =(D'A_nD)^{-1}D'A_n\widehat b(q).
\label{eq:generic-md-closed}
\end{equation}
The efficient minimum-distance estimator within the class based on $\widehat b(q)$ uses $A_n=\widehat\Omega_q^{-1}$:
\begin{equation}
\widehat\beta_{\mathrm{MD}}(q)
 =(D'\widehat\Omega_q^{-1}D)^{-1}
 D'\widehat\Omega_q^{-1}\widehat b(q).
\label{eq:optimal-md}
\end{equation}
Here efficiency is relative to regular linear combinations of the first-step quantile-projection restrictions; it is not a claim of semiparametric efficiency for the full distributional model.

\begin{proposition}[Consistency]
\label{prop:consistency}
Suppose Assumptions \ref{ass:sampling}--\ref{ass:unique-projection} hold, and let each $\widehat\pi_t(q)$ be any measurable minimizer of \eqref{eq:sample-first-stage}. Then
\begin{equation}
\widehat\Pi(q)\toP\Pi_0(q),
\qquad
\widehat b(q)\toP D\beta_0.
\label{eq:first-step-consistency}
\end{equation}
If $A_n\toP A$ for a positive-definite matrix $A$, then
\begin{equation}
\widehat\beta_A(q)\toP\beta_0.
\label{eq:md-consistency}
\end{equation}
In particular, both the equal-weight and feasible efficient minimum-distance estimators are consistent whenever the estimated weight converges to a positive-definite limit.
\end{proposition}

\subsection{Equivalent full-system formulation}

The estimator can also be written in the full form suggested by \eqref{eq:projection-representation}. Let
\begin{equation}
G=\begin{pmatrix}E_1\\E_2\\ \vdots\\E_T\end{pmatrix}
 \in\R^{pT^2\times p},
\qquad
H=\1_T\otimes I_{pT}
 \in\R^{pT^2\times pT},
\label{eq:G-H}
\end{equation}
and let $M=[G\ H]$. Then
\begin{equation}
\Pi_0(q)=G\beta_0+H\delta_0(q)=M\theta_0(q),
\qquad
\theta_0(q)=\begin{pmatrix}\beta_0\\ \delta_0(q)\end{pmatrix}.
\label{eq:full-system}
\end{equation}
For $T\geq 2$, $M$ has full column rank $p+pT$. If $\widehat\Sigma_{\Pi,q}$ consistently estimates the covariance of $\sqrt n\{\widehat\Pi(q)-\Pi_0(q)\}$, full-system GLS is
\begin{equation}
\widehat\theta_{\mathrm{GLS}}(q)
 =\{M'\widehat\Sigma_{\Pi,q}^{-1}M\}^{-1}
 M'\widehat\Sigma_{\Pi,q}^{-1}\widehat\Pi(q).
\label{eq:full-gls}
\end{equation}

\begin{proposition}[Equivalence of the two formulations]
\label{prop:equivalence}
Suppose $T\geq 2$ and $\Sigma_{\Pi,q}$ is positive definite. Let
\[
\Omega_q=C\Sigma_{\Pi,q}C'.
\]
Then the first $p$ coordinates of the population-weight version of \eqref{eq:full-gls} are identical to the contrast estimator in \eqref{eq:optimal-md} with weight $\Omega_q^{-1}$. The same equivalence holds with consistent estimated covariance matrices up to $o_p(n^{-1/2})$.
\end{proposition}

The contrast formulation is often simpler computationally because it eliminates the nuisance projection coefficient before the second step. The full-system formulation is useful for seeing all model restrictions and for recovering $\delta_0(q)$ if desired.

\subsection{Implementation}

The estimator can be implemented as follows.

\begin{enumerate}[label=\textbf{Step \arabic*.},leftmargin=5.5em]
\item Form $X_i=(X_{i1}',\ldots,X_{iT}')'$.

\item For each $t$, run the cross-sectional quantile regression in \eqref{eq:sample-first-stage} and retain the full vector $\widehat\pi_t(q)$.

\item For every ordered pair $s\neq t$, extract the coefficient on $X_{is}$ from the period-$s$ and period-$t$ regressions and compute \eqref{eq:sample-contrast}.

\item Compute the equal-weight estimator \eqref{eq:equal-weight-estimator}. This provides a consistent preliminary estimate and a convenient diagnostic display of the individual contrasts.

\item Estimate the joint covariance $\Omega_q$ using Section \ref{subsec:feasible-covariance}, and compute \eqref{eq:optimal-md}.

\item Report the covariance estimator in \eqref{eq:beta-vhat}, the overidentification statistic in \eqref{eq:J-statistic}, and, when desired, cluster-bootstrap confidence intervals.
\end{enumerate}

\section{Limit theory and inference}
\label{sec:asymptotics}

\subsection{Joint asymptotic distribution of the first step}

Define the population projection residual
\begin{equation}
U_{it}(q)=Y_{it}-X_i'\pi_{0t}(q)
 =A_i+V_{it}-X_i'\delta_0(q).
\label{eq:projection-residual}
\end{equation}
Under stationarity, $U_{it}(q)\mid X_i$ has the same distribution for every $t$. Let $f_{U(q)\mid X}(u\mid X_i)$ denote its conditional density.

\begin{assumption}[Regularity for quantile projections]
\label{ass:qr-regularity}
For the chosen $q$:
\begin{enumerate}[label=(\alph*)]
\item $\E\|X_i\|^{2+\eta}<\infty$ for some $\eta>0$.

\item The conditional density $f_{U(q)\mid X}(u\mid x)$ exists in a neighborhood of zero, is uniformly bounded there, and is continuous at zero for almost every $x$.

\item The matrix
\begin{equation}
J_q=\E\big[f_{U(q)\mid X}(0\mid X_i)X_iX_i'\big]
\label{eq:Jq}
\end{equation}
is finite and positive definite.

\item The covariance matrices defined below are finite, and $\Omega_q$ is positive definite.
\end{enumerate}
\end{assumption}

Continuity at zero and the first-order condition for the population projection imply
\begin{equation}
\E\left[X_i\psi_q\{U_{it}(q)\}\right]=0,
\qquad t=1,\ldots,T.
\label{eq:population-score-zero}
\end{equation}
For each individual define the stacked quantile score
\begin{equation}
g_i(q)
 =\begin{pmatrix}
 X_i\psi_q(U_{i1}(q))\\
 \vdots\\
 X_i\psi_q(U_{iT}(q))
 \end{pmatrix}
 \in\R^{pT^2}.
\label{eq:stacked-score}
\end{equation}
Let
\begin{equation}
S_q=\E[g_i(q)g_i(q)'],
\qquad
\mathcal J_q=I_T\otimes J_q.
\label{eq:S-Jbig}
\end{equation}
The off-diagonal blocks of $S_q$ retain arbitrary within-individual dependence. In particular, for periods $t$ and $s$,
\begin{equation}
S_{q,ts}
 =\E\left[
 X_iX_i'\psi_q(U_{it}(q))\psi_q(U_{is}(q))
 \right].
\label{eq:score-cross-cov}
\end{equation}

\begin{theorem}[Joint Bahadur representation]
\label{thm:first-stage-asymptotic}
Under Assumptions \ref{ass:sampling}--\ref{ass:qr-regularity},
\begin{equation}
\sqrt n\{\widehat\Pi(q)-\Pi_0(q)\}
 =\mathcal J_q^{-1}\frac{1}{\sqrt n}
 \sum_{i=1}^n g_i(q)+o_p(1),
\label{eq:bahadur-stacked}
\end{equation}
and therefore
\begin{equation}
\sqrt n\{\widehat\Pi(q)-\Pi_0(q)\}
 \toD N(0,\Sigma_{\Pi,q}),
\qquad
\Sigma_{\Pi,q}=\mathcal J_q^{-1}S_q\mathcal J_q^{-1}.
\label{eq:SigmaPi}
\end{equation}
Consequently,
\begin{equation}
\sqrt n\{\widehat b(q)-D\beta_0\}
 \toD N(0,\Omega_q),
\qquad
\Omega_q=C\Sigma_{\Pi,q}C'.
\label{eq:Omegaq}
\end{equation}
\end{theorem}

The theorem uses the misspecification-robust quantile-regression covariance. In general,
\[
\Pp\{U_{it}(q)\leq 0\mid X_i\}\neq q,
\]
because $X_i'\pi_{0t}(q)$ is a linear projection rather than the true conditional quantile. Thus the score covariance in \eqref{eq:score-cross-cov} cannot be replaced mechanically by $q(1-q)\E[X_iX_i']$, and cross-period blocks cannot be ignored.

\begin{corollary}[Asymptotic distribution of minimum distance]
\label{cor:md-asymptotic}
Let $A_n\toP A$, where $A$ is positive definite. Then
\begin{equation}
\sqrt n\{\widehat\beta_A(q)-\beta_0\}
 \toD N(0,V_A),
\label{eq:generic-md-limit}
\end{equation}
where
\begin{equation}
V_A
 =(D'AD)^{-1}D'A\Omega_q A D(D'AD)^{-1}.
\label{eq:generic-md-var}
\end{equation}
For the optimal weight $A=\Omega_q^{-1}$,
\begin{equation}
V_{\mathrm{MD},q}
 =(D'\Omega_q^{-1}D)^{-1}.
\label{eq:optimal-md-var}
\end{equation}
\end{corollary}

\subsection{Feasible analytic covariance estimation}
\label{subsec:feasible-covariance}

Let
\begin{equation}
\widehat U_{it}(q)=Y_{it}-X_i'\widehat\pi_t(q).
\label{eq:residual-hat}
\end{equation}
Choose a bounded, symmetric, Lipschitz kernel $K$ satisfying $\int K(v)\,dv=1$, and a bandwidth $h_n$. A period-specific estimator of the quantile-regression Jacobian is
\begin{equation}
\widehat J_{q,t}
 =\frac{1}{nh_n}\sum_{i=1}^n
 K\!\left(\frac{\widehat U_{it}(q)}{h_n}\right)X_iX_i'.
\label{eq:Jhat-period}
\end{equation}
Because the population Jacobian is common across periods under stationarity, a pooled estimator is
\begin{equation}
\widehat J_q=\frac{1}{T}\sum_{t=1}^T\widehat J_{q,t},
\qquad
\widehat{\mathcal J}_q=I_T\otimes\widehat J_q.
\label{eq:Jhat-pooled}
\end{equation}
Using separate $\widehat J_{q,t}$ on the block diagonal is also consistent and can be useful as a diagnostic for misspecification.

Define the estimated stacked score
\begin{equation}
\widehat g_i(q)
 =\begin{pmatrix}
 X_i\psi_q(\widehat U_{i1}(q))\\
 \vdots\\
 X_i\psi_q(\widehat U_{iT}(q))
 \end{pmatrix},
\qquad
\overline g(q)=\frac{1}{n}\sum_{i=1}^n\widehat g_i(q),
\label{eq:ghat}
\end{equation}
and estimate its covariance by the cluster outer product
\begin{equation}
\widehat S_q
 =\frac{1}{n}\sum_{i=1}^n
 \{\widehat g_i(q)-\overline g(q)\}
 \{\widehat g_i(q)-\overline g(q)\}'.
\label{eq:Shat}
\end{equation}
Then
\begin{equation}
\widehat\Sigma_{\Pi,q}
 =\widehat{\mathcal J}_q^{-1}
 \widehat S_q
 \widehat{\mathcal J}_q^{-1},
\qquad
\widehat\Omega_q
 =C\widehat\Sigma_{\Pi,q}C'.
\label{eq:Sigma-Omega-hat}
\end{equation}
The estimated covariance of $\sqrt n\{\widehat\beta_{\mathrm{MD}}(q)-\beta_0\}$ is
\begin{equation}
\widehat V_{\mathrm{MD},q}
 =(D'\widehat\Omega_q^{-1}D)^{-1},
\label{eq:beta-vhat}
\end{equation}
and the estimated covariance of $\widehat\beta_{\mathrm{MD}}(q)$ itself is $\widehat V_{\mathrm{MD},q}/n$.

\begin{assumption}[Kernel and bandwidth]
\label{ass:kernel}
In addition to Assumption \ref{ass:qr-regularity}, suppose $\E\|X_i\|^{4+\eta}<\infty$ for some $\eta>0$, the conditional density is locally Lipschitz at zero, $K$ satisfies the conditions above, and
\begin{equation}
h_n\to 0,
\qquad nh_n\to\infty,
\qquad \sqrt n\,h_n^2\to\infty.
\label{eq:bandwidth}
\end{equation}
\end{assumption}

\begin{proposition}[Consistency of the analytic covariance estimator]
\label{prop:covariance-consistency}
Under Assumptions \ref{ass:sampling}--\ref{ass:kernel},
\begin{equation}
\widehat J_q\toP J_q,
\quad
\widehat S_q\toP S_q,
\quad
\widehat\Sigma_{\Pi,q}\toP\Sigma_{\Pi,q},
\quad
\widehat\Omega_q\toP\Omega_q,
\label{eq:all-cov-consistency}
\end{equation}
and
\begin{equation}
\widehat V_{\mathrm{MD},q}\toP V_{\mathrm{MD},q}.
\label{eq:V-consistency}
\end{equation}
\end{proposition}

A direct influence-function implementation avoids storing the full $pT^2\times pT^2$ covariance matrix. Set
\[
\widehat\varphi_i(q)
 =C\widehat{\mathcal J}_q^{-1}
 \{\widehat g_i(q)-\overline g(q)\}.
\]
Then $n^{-1}\sum_i\widehat\varphi_i(q)\widehat\varphi_i(q)'$ is algebraically equal to $\widehat\Omega_q$ in \eqref{eq:Sigma-Omega-hat}.

\subsection{Overidentification test}

The model requires all $r=T(T-1)$ contrast vectors to equal the same $\beta_0$. With the efficient estimator, define
\begin{equation}
J_n(q)
 =n\{\widehat b(q)-D\widehat\beta_{\mathrm{MD}}(q)\}'
 \widehat\Omega_q^{-1}
 \{\widehat b(q)-D\widehat\beta_{\mathrm{MD}}(q)\}.
\label{eq:J-statistic}
\end{equation}

\begin{proposition}[Projection-restriction test]
\label{prop:J-test}
Under the assumptions of Proposition \ref{prop:covariance-consistency} and the stationary location model,
\begin{equation}
J_n(q)\toD\chi^2_{p\{T(T-1)-1\}}.
\label{eq:J-limit}
\end{equation}
For $T=2$, the degrees of freedom are $p$.
\end{proposition}

The test is a specification test for the finite collection of quantile-projection restrictions in \eqref{eq:all-contrasts-beta}. Rejection is evidence against the stationary location model or against the maintained linear location specification. Nonrejection does not establish equality of the full conditional disturbance distributions, because the test uses only their linear quantile projections.

\subsection{Pairs cluster bootstrap}

The analytic covariance estimator requires a density estimate at the fitted quantile projection. A simple alternative is a pairs bootstrap at the individual level. In each bootstrap replication, draw $n$ individuals with replacement, retaining the full $T$-period vector for each selected individual. Re-estimate the $T$ first-step quantile regressions and recompute the minimum-distance estimator. This preserves all within-individual dependence.

Let $\widehat\beta_{\mathrm{MD}}^{*(b)}(q)$ denote the result in bootstrap replication $b$, and let $\overline\beta^*(q)=B^{-1}\sum_{b=1}^B\widehat\beta_{\mathrm{MD}}^{*(b)}(q)$. The bootstrap covariance estimator is
\begin{equation}
\widehat{\Var}^{*}\{\widehat\beta_{\mathrm{MD}}(q)\}
 =\frac{1}{B-1}\sum_{b=1}^B
 \{\widehat\beta_{\mathrm{MD}}^{*(b)}(q)-\overline\beta^*(q)\}
 \{\widehat\beta_{\mathrm{MD}}^{*(b)}(q)-\overline\beta^*(q)\}'.
\label{eq:bootstrap-cov}
\end{equation}
Under the regularity conditions for Theorem \ref{thm:first-stage-asymptotic}, the conditional distribution of
$\sqrt n\{\widehat\beta_{\mathrm{MD}}^*(q)-\widehat\beta_{\mathrm{MD}}(q)\}$ consistently estimates the limiting distribution in Corollary \ref{cor:md-asymptotic}. The weight matrix may be re-estimated in each bootstrap sample or held fixed at a consistent full-sample estimate; the two choices are first-order equivalent.

\subsection{Covariance matrices of observable residuals}
\label{subsec:residual-covariance}

It is useful to distinguish three covariance objects. The covariance needed for inference on $\beta_0$ is the quantile-score covariance $S_q$ in \eqref{eq:S-Jbig}. The covariance of the projection residual vector $U_i(q)=(U_{i1}(q),\ldots,U_{iT}(q))'$ is a different object and can be estimated directly by the sample covariance of $\widehat U_i(q)$. Finally, the structural composite disturbance
\begin{equation}
R_i(\beta_0)
 =\begin{pmatrix}
 Y_{i1}-X_{i1}'\beta_0\\
 \vdots\\
 Y_{iT}-X_{iT}'\beta_0
 \end{pmatrix}
 =A_i\1_T+V_i
\label{eq:composite-vector}
\end{equation}
is observable up to $\beta_0$.

Suppose $\E\|R_i(\beta_0)\|^2<\infty$. Define
\begin{equation}
\widehat R_i
 =R_i\{\widehat\beta_{\mathrm{MD}}(q)\},
\qquad
\overline R=\frac{1}{n}\sum_{i=1}^n\widehat R_i,
\label{eq:Rhat}
\end{equation}
and
\begin{equation}
\widehat\Gamma_R
 =\frac{1}{n}\sum_{i=1}^n
 (\widehat R_i-\overline R)(\widehat R_i-\overline R)'.
\label{eq:GammaRhat}
\end{equation}
Then $\widehat\Gamma_R\toP\Gamma_R:=\Var\{R_i(\beta_0)\}$. An analogous statement holds for the sample covariance of $\widehat U_i(q)$. These plug-in estimators are useful for describing serial dependence in the observable composite disturbance.

The relation between $\Gamma_R$ and the covariance of the latent idiosyncratic disturbance can be stated precisely. Let $V_i=(V_{i1},\ldots,V_{iT})'$ and $\Gamma_V=\Var(V_i)$.

\begin{proposition}[What is identified about the covariance of $V_i$]
\label{prop:latent-covariance}
Suppose $\E A_i^2+\E\|V_i\|^2<\infty$ and Assumption \ref{ass:stationarity} holds. Then $\Cov(A_i,V_{it})$ is common across $t$, and there is a scalar $\lambda_0$ such that
\begin{equation}
\Gamma_R=\Gamma_V+\lambda_0\1_T\1_T'.
\label{eq:covariance-decomposition}
\end{equation}
Consequently, for every fixed matrix $L$ satisfying $L\1_T=0$,
\begin{equation}
\Var(LV_i)=L\Gamma_VL'=L\Gamma_RL',
\label{eq:difference-covariance}
\end{equation}
and $L\widehat\Gamma_RL'$ is consistent for this covariance. In particular, the covariance matrix of any vector of time differences of $V_i$ is point identified. The level covariance $\Gamma_V$ is identified only up to the scalar component $\lambda_0\1_T\1_T'$ under the baseline assumptions.
\end{proposition}

Thus the covariance matrix of the latent $V_i$ is not separately identified from the persistent component generated by $A_i$. Orthogonality between $A_i$ and $V_i$ alone does not resolve the problem because $\Var(A_i)$ remains unknown. A restriction on at least one serial covariance of $V_i$, serial independence, or another normalization that identifies $\lambda_0$ would complete the decomposition. None of these additional restrictions is required for inference on $\beta_0$.

\section{Combining several quantile projections}
\label{sec:multiple-q}

Stationarity implies \eqref{eq:projection-representation} for every $q$ at which the relevant quantile projection is unique. Therefore several quantile indices can be combined without changing the target parameter. Let $0<q_1<\cdots<q_L<1$ be fixed. For each $q_\ell$, construct $\widehat b(q_\ell)$ as in \eqref{eq:stacked-b}, and stack
\begin{equation}
\widehat b_L
 =\begin{pmatrix}
 \widehat b(q_1)\\ \vdots\\ \widehat b(q_L)
 \end{pmatrix},
\qquad
D_L=\1_L\otimes D.
\label{eq:multi-q-stack}
\end{equation}
The population restriction is
\begin{equation}
b_{0,L}=D_L\beta_0.
\label{eq:multi-q-restriction}
\end{equation}
Let $\Omega_L$ be the covariance of $\sqrt n(\widehat b_L-b_{0,L})$, including cross-quantile covariances. These covariances are estimated by stacking the influence functions based on
\[
X_i\psi_{q_\ell}\{U_{it}(q_\ell)\},
\qquad \ell=1,\ldots,L,
\quad t=1,\ldots,T.
\]
The composite minimum-distance estimator is
\begin{equation}
\widehat\beta_{\mathrm{CMD}}
 =(D_L'\widehat\Omega_L^{-1}D_L)^{-1}
 D_L'\widehat\Omega_L^{-1}
 \widehat b_L,
\label{eq:composite-md}
\end{equation}
with asymptotic variance
\begin{equation}
V_{\mathrm{CMD}}
 =(D_L'\Omega_L^{-1}D_L)^{-1}.
\label{eq:composite-var}
\end{equation}
This estimator can improve precision when the quantile-projection scores contain complementary information. The use of cross-quantile restrictions follows the general GMM and minimum-distance logic in \citet{firpoetal2022}, here applied to the common-slope restrictions generated by panel stationarity. It also yields a joint specification test with $p\{LT(T-1)-1\}$ degrees of freedom. Rejection of the cross-quantile restrictions is particularly informative in light of Theorem \ref{thm:incompatibility}: it indicates that the data do not support a single stationary location coefficient across the chosen quantile projections.


\section{Monte Carlo experiments}\label{mc:sec}

This section reports simulation evidence on the finite-sample behavior of
the estimators and specification tests developed above. The experiments
are organized around five questions. First, how do the equally weighted
and efficient minimum-distance estimators perform in a correctly
specified baseline, and how much do they gain over a single
diagonal-minus-off-diagonal contrast? Second, how does performance
deteriorate as within-individual variation in the regressor shrinks, and
how much of that deterioration is offset by additional time periods?
Third, does the misspecification-robust covariance matrix deliver
reliable inference when the period-specific quantile projections are not
correct conditional quantile functions, and how badly do conventional
quantile-regression standard errors fail in that case? Fourth, what is
the efficiency cost of the quantile-based procedure under Gaussian errors
and its benefit under heavy tails, relative to least-squares fixed-effects
estimation? Fifth, do the overidentification tests detect violations of
conditional time stationarity, and does the cross-quantile test have
power against quantile-varying slopes that single-quantile tests cannot
detect, as Theorem~\ref{thm:incompatibility} predicts?%

\subsection{Common design and implementation}\label{mc:common}

The baseline data-generating process is a one-regressor version of
model~\eqref{eq:location-model} in which all maintained assumptions hold. Let $C_i$, $Z_{it}$, $\eta_i$, and $\varepsilon_{it}$ be
mutually independent standard normal variables. The regressor is
\begin{equation}
X_{it}=\frac{C_i+\sigma_X Z_{it}}{\sqrt{1+\sigma_X^{2}}},
\qquad t=1,\dots,T,
\label{mc:eq:X}
\end{equation}
so that $\mathrm{Var}(X_{it})=1$, the within-individual correlation is
$1/(1+\sigma_X^{2})$, and $\mathrm{Var}(X_{it}-X_{is})=
2\sigma_X^{2}/(1+\sigma_X^{2})$ for $s\neq t$. The individual effect is
correlated with the regressor history through
\begin{equation}
A_i=\rho_A\,\overline X_i/s_X+\sqrt{1-\rho_A^{2}}\,\eta_i,
\qquad
s_X=\mathrm{sd}(\overline X_i)=
\sqrt{\tfrac{1+\sigma_X^{2}/T}{1+\sigma_X^{2}}},
\label{mc:eq:A}
\end{equation}
and the idiosyncratic error $V_{it}$ is a stationary Gaussian AR(1),
$V_{it}=\rho_V V_{i,t-1}+\sqrt{1-\rho_V^{2}}\,\varepsilon_{it}$ with
$V_{i1}\sim N(0,1)$, independent of $(X_i,A_i)$, so conditional time
stationarity holds by construction. Outcomes are
$Y_{it}=\beta_0 X_{it}+A_i+V_{it}$ with $\beta_0=1$. Because $(X_i,A_i,V_i)$
is jointly Gaussian in the baseline, the period-$t$ cross-sectional
quantile projections of $Y_{it}$ on the full history
$X_i=(1,X_{i1},\dots,X_{iT})'$ are exact conditional quantile functions,
with slope $\beta_0+\rho_A/(Ts_X)$ on $X_{it}$ and $\rho_A/(Ts_X)$ on
$X_{is}$, $s\neq t$; the diagonal-minus-off-diagonal contrasts therefore
identify $\beta_0$ at every quantile, and closed-form expressions for the
asymptotic covariance matrix $\Omega_q$ of the stacked contrasts are
available. We use these expressions to construct an infeasible
``oracle'' minimum-distance estimator, denoted OMD, that weights with
$\Omega_q^{-1}$ evaluated at the truth; it provides the efficiency
benchmark against which the feasible two-step estimator can be judged.

Every design is estimated the same way. The first step runs $T$
cross-sectional linear quantile regressions of $Y_{it}$ on $X_i$, one per
period, each including a common intercept that is excluded from the
structural contrasts. Quantile regressions are solved with a
Frisch--Newton interior-point algorithm; solutions agree with exact
linear-programming solutions to at least six decimal places, which
matters because solver noise of even $10^{-2}$ would contaminate the
size of the overidentification tests. The second step forms the
$r=T(T-1)$ contrasts $\widehat b_{st}(q)$ for all ordered pairs
$s\neq t$ and computes: the single-contrast estimator
$\widehat\beta_{SC}$ based on $(s,t)=(1,2)$; the equally weighted
average $\widehat\beta_{EW}$; and the minimum-distance estimator
$\widehat\beta_{MD}$ with weight $\widehat\Omega_q^{-1}$, together with
the overidentification statistic $J_n(q)$ on $r-1$ degrees of freedom.
The covariance matrix $\widehat\Omega_q$ is the misspecification-robust
cluster sandwich of Section~\ref{sec:asymptotics}: the Jacobian is estimated by a
Gaussian-kernel density estimate at zero of the projection residuals,
pooled across periods, with the Silverman rule-of-thumb bandwidth
$h=1.06\,\widehat\sigma\,n^{-1/5}$ (a robust scale estimate,
sensitivity to $h$ is examined below), and the score outer product is
clustered at the individual level, retaining all cross-period blocks.
Inverses are computed directly unless the condition number exceeds
$10^{12}$, in which case a pseudoinverse is substituted and the event
recorded; in the designs of this section such adjustments essentially
never occur, and we report their frequency. Composite cross-quantile
estimation stacks the contrasts for $L$ quantiles and weights with the
inverse of the corresponding $rL\times rL$ cluster covariance matrix;
the associated joint statistic has $rL-1$ degrees of freedom. Unless
noted otherwise each configuration uses $2{,}000$ Monte Carlo
replications; bootstrap experiments use $1{,}000$ replications with
$B=399$ pairs-bootstrap draws. Coverage refers to nominal $95\%$
equal-tailed intervals $\widehat\beta\pm1.96\,\mathrm{se}$, and all test
rejection rates are at the $5\%$ level.

\subsection{Design 1: baseline performance}\label{mc:d1}

The first design crosses $\rho_A\in\{0,0.75\}$, $\rho_V\in\{0,0.7\}$,
$n\in\{250,500,1000\}$, and $T\in\{2,4\}$ at the median, $q=0.5$, with
$\sigma_X=1$. Table~\ref{mc:tab:d1} reports the empirically relevant
half of the grid with a correlated individual effect
($\rho_A=0.75$); the full grid, including the exogenous case
$\rho_A=0$ whose results are nearly identical, appears in
Table~\ref{mc:tab:d1full}.

\begin{table}[t]
\centering
\caption{Design 1: baseline with correlated individual effect
($\rho_A=0.75$, $\sigma_X=1$, $q=0.5$, $2{,}000$ replications). SC is
the single contrast $\widehat b_{12}$; EW and MD are the equally
weighted and feasible minimum-distance estimators; OMD uses the oracle
weight. ``se'' is the average robust standard error, ``cov'' the
coverage of nominal $95\%$ intervals, and $J$ the rejection rate of
$J_n(0.5)$ at the $5\%$ level ($r-1$ degrees of freedom). The full
grid including $\rho_A=0$ is in Table~\ref{mc:tab:d1full}.}
\label{mc:tab:d1}
\setlength{\tabcolsep}{3.4pt}
{\small
\begin{tabular}{llcccccccccccc}
\toprule
&& SC & \multicolumn{4}{c}{EW} & \multicolumn{5}{c}{MD} & OMD & \\
\cmidrule(lr){3-3}\cmidrule(lr){4-7}\cmidrule(lr){8-12}\cmidrule(lr){13-13}
Configuration & $n$ & rmse & bias & sd & rmse & cov & bias & sd & rmse
& se & cov & rmse & $J$ \\
\midrule
$T=2$, $\rho_V=0$ & 250 & 0.143 & 0.002 & 0.125 & 0.125 & 0.950 & 0.002 & 0.125 & 0.125 & 0.123 & 0.946 & 0.125 & 0.038 \\
 & 500 & 0.098 & -0.002 & 0.084 & 0.084 & 0.963 & -0.003 & 0.085 & 0.085 & 0.087 & 0.960 & 0.084 & 0.052 \\
 & 1000 & 0.070 & 0.002 & 0.060 & 0.060 & 0.955 & 0.002 & 0.060 & 0.060 & 0.061 & 0.957 & 0.060 & 0.046 \\
\addlinespace
$T=2$, $\rho_V=0.7$ & 250 & 0.099 & 0.000 & 0.087 & 0.087 & 0.949 & 0.000 & 0.088 & 0.088 & 0.089 & 0.945 & 0.087 & 0.045 \\
 & 500 & 0.072 & -0.000 & 0.062 & 0.062 & 0.954 & -0.000 & 0.062 & 0.062 & 0.063 & 0.954 & 0.062 & 0.046 \\
 & 1000 & 0.051 & -0.000 & 0.044 & 0.044 & 0.958 & -0.000 & 0.044 & 0.044 & 0.044 & 0.955 & 0.044 & 0.050 \\
\addlinespace
$T=4$, $\rho_V=0$ & 250 & 0.158 & -0.001 & 0.070 & 0.070 & 0.951 & -0.002 & 0.073 & 0.073 & 0.067 & 0.930 & 0.070 & 0.088 \\
 & 500 & 0.106 & 0.001 & 0.050 & 0.050 & 0.949 & 0.001 & 0.050 & 0.050 & 0.049 & 0.944 & 0.050 & 0.060 \\
 & 1000 & 0.077 & -0.001 & 0.035 & 0.035 & 0.949 & -0.001 & 0.035 & 0.035 & 0.035 & 0.941 & 0.035 & 0.052 \\
\addlinespace
$T=4$, $\rho_V=0.7$ & 250 & 0.110 & 0.000 & 0.057 & 0.057 & 0.945 & 0.001 & 0.058 & 0.058 & 0.050 & 0.913 & 0.054 & 0.099 \\
 & 500 & 0.077 & 0.001 & 0.039 & 0.039 & 0.956 & 0.001 & 0.038 & 0.038 & 0.037 & 0.937 & 0.038 & 0.065 \\
 & 1000 & 0.055 & 0.000 & 0.027 & 0.027 & 0.957 & 0.000 & 0.026 & 0.026 & 0.026 & 0.953 & 0.026 & 0.061 \\
\addlinespace
\bottomrule
\end{tabular}}
\end{table}

Four findings stand out. First, all estimators are essentially unbiased
at every sample size: no mean bias in the table exceeds
$0.005$ in absolute value, consistent with the identification result
that each contrast recovers $\beta_0$ exactly. Second, averaging over
contrasts is valuable. With $T=2$ the two contrasts are close to
exchangeable and $\widehat\beta_{EW}$ reduces the RMSE of the single
contrast by roughly $12$--$14\%$; with $T=4$, where twelve contrasts are
available, the reduction is roughly $50\%$ at every $n$.
Third, the feasible minimum-distance estimator is never appreciably
worse than the equally weighted estimator and the scope for efficiency
gains in this design is modest: the oracle weighting reduces the
variance relative to equal weighting by about $8\%$ when $T=4$ and
$\rho_V=0.7$ (variance ratio $0.92$) and by nothing when $T=2$, since
with two symmetric contrasts equal weighting is already efficient. The
feasible estimator captures the oracle gain only as $n$ grows: at
$n=250$ estimating the $12\times12$ weight matrix costs slightly more
than optimal weighting gains, while by $n=1000$ the feasible and oracle
RMSEs agree to three decimals. Fourth, inference is reliable. Coverage
of the robust intervals lies between $0.94$ and $0.96$ in every
configuration for $\widehat\beta_{EW}$, with mild undercoverage for
$\widehat\beta_{MD}$ at $T=4$, $n=250$ (about $0.91$--$0.93$) that
disappears as $n$ grows, and the $J_n(0.5)$ statistic is close to its
nominal size for $T=2$ at all $n$ (rejection rates $0.038$--$0.061$)
while showing the familiar modest overrejection of overidentification
tests with estimated weight matrices at $T=4$ and small $n$
($0.088$--$0.112$ at $n=250$, declining to $0.050$--$0.061$ at
$n=1000$). Average robust standard errors track the Monte Carlo standard
deviations closely throughout.

\subsection{Design 2: weak within-individual variation and the role of
additional periods}\label{mc:d2}

Identification rests entirely on within-individual movement in the
regressor, so the practically important stress test shrinks
$\sigma_X$. We fix $\rho_A=0.75$, $\rho_V=0.5$, $q=0.5$, and cross
$\sigma_X\in\{1,0.5,0.25\}$---equivalently
$\mathrm{Var}(X_{it}-X_{is})\in\{1,0.4,0.118\}$---with
$T\in\{2,4,6\}$ and $n\in\{500,1000\}$. Figure~\ref{mc:fig:d2}
summarizes the results at $n=1000$; the complete numerical results are
in Table~\ref{mc:tab:d2}.

\begin{figure}[t]
\centering
\includegraphics[width=\textwidth]{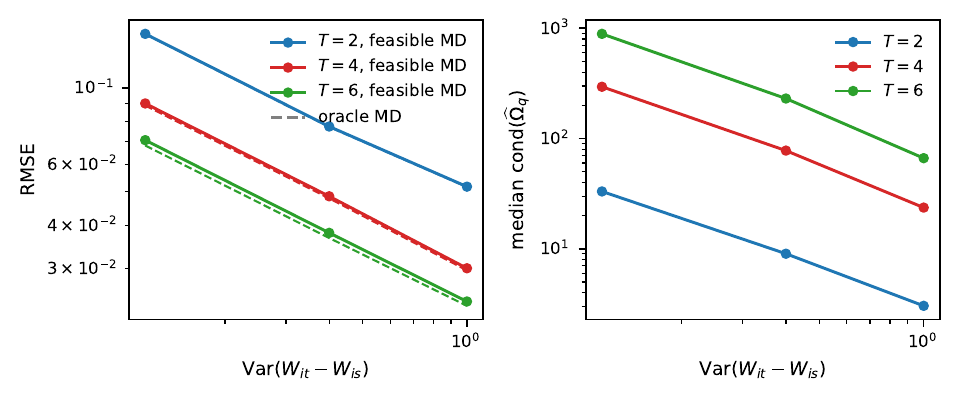}
\caption{Design 2. Left: RMSE of the feasible minimum-distance estimator
(solid, markers) and the oracle-weighted estimator (dashed) at $n=1000$,
$q=0.5$, plotted against the within-pair regressor variance
$\mathrm{Var}(X_{it}-X_{is})=2\sigma_X^2/(1+\sigma_X^2)$; both axes are
logarithmic. The equally weighted estimator is visually
indistinguishable from feasible MD and is omitted. Right: median
condition number of $\widehat\Omega_{q}$ across replications.}
\label{mc:fig:d2}
\end{figure}

Three conclusions emerge. First, precision deteriorates at the expected
rate but the procedure remains entirely well behaved: moving from
$\sigma_X=1$ to $\sigma_X=0.25$ multiplies the RMSE by a factor of $2.8$--$3.0$ at every $(n,T)$,
close to the ratio $\sqrt{1/0.118}\approx2.9$ implied by the reduction in usable
variation, bias remains negligible, and coverage stays within sampling
error of $0.95$. Second, additional periods substitute for
within-period variation: at $\sigma_X=0.25$ and $n=500$, going from
$T=2$ to $T=6$ cuts the RMSE from $0.203$ to $0.103$, so three times as
many periods roughly halve the RMSE, in line with the increase in the
number of contrasts from $2$ to $30$. Third, the numerical warnings that
motivate the diagnostics are visible but benign in these ranges. The
median condition number of $\widehat\Omega_q$ rises steeply, from about
$3$ at $(\sigma_X,T)=(1,2)$ to about $1{,}040$ at $(0.25,6)$, yet it
remains far from the $10^{12}$ threshold: the pseudoinverse safeguard
was never triggered in $36{,}000$ replications. The cost of weight-matrix
estimation does, however, show up in the tests: with $T=6$ (thirty
contrasts) and $n=500$ the $J_n$ statistic rejects a true null $11$--$13\%$
of the time, falling to $7$--$9\%$ at $n=1000$, and the efficient
estimator undercovers slightly (about $0.92$--$0.93$) where equal
weighting remains at $0.95$--$0.96$. When the number of contrasts is
large relative to $n$, the equally weighted estimator paired with robust
standard errors is the more dependable summary, and the oracle
calculations show the forgone efficiency is below $9\%$ of variance in
this design.

\subsection{Design 3: misspecified projections and robust
inference}\label{mc:d3}

The robust covariance matrix is designed to remain valid when the
period-specific linear quantile projections are not correct conditional
quantile functions. To generate that situation while preserving the
identification restrictions, we keep \eqref{mc:eq:X} with $\sigma_X=1$
and $\rho_V=0.5$ but make the individual effect nonlinear in the
history and the error scale dependent on it:
\begin{equation}
A_i=0.75\,\widetilde X_i+0.5\,(\widetilde X_i^{2}-1)+\eta_i,
\qquad
V_{it}=\bigl(0.5+0.5\,|\widetilde X_i|\bigr) G_{it},
\qquad
\widetilde X_i=\overline X_i/s_X,
\label{mc:eq:d3}
\end{equation}
where $G_{it}$ is the Gaussian AR(1) above. Conditional time
stationarity still holds, so every contrast continues to estimate
$\beta_0$, but the true conditional quantile of $Y_{it}$ given $X_i$ is
nonlinear in $X_i$ and the linear projections are genuinely
misspecified. We compare four ways of assessing the sampling
variability of the same point estimator $\widehat\beta_{EW}$: (i) the
full robust sandwich, with cross-period score covariances; (ii) a
restricted sandwich that keeps the robust per-period blocks but zeroes
the cross-period blocks; (iii) conventional quantile-regression standard
errors that use the $q(1-q)E[XX']$ score matrix period by period; and
(iv) the pairs cluster bootstrap. Panel A of Table~\ref{mc:tab:d3}
reports the analytic methods across
$q\in\{0.25,0.5,0.75\}$, $T\in\{2,4\}$, $n\in\{250,500,1000\}$; Panel B
reports the bootstrap experiment at $T=2$, $n=500$.

\begin{table}[t]
\centering
\caption{Design 3: nonlinear conditional quantiles. Panel A compares
standard errors for the same estimator $\widehat\beta_{EW}(q)$: the full
robust sandwich (``full''), the sandwich with cross-period score blocks
set to zero (``no cross''), and conventional per-period
quantile-regression standard errors (``conv.''); each pair of columns
reports the average standard error and the resulting coverage. MD cov is
the coverage of the feasible minimum-distance estimator with the full
sandwich, and $J$ the rejection rate of $J_n(q)$. Panel B reports the
pairs cluster bootstrap ($T=2$, $n=500$, $1{,}000$ replications,
$B=399$, minimum-distance weight fixed at the full-sample estimate).}
\label{mc:tab:d3}
\setlength{\tabcolsep}{3.2pt}
{\small
\begin{tabular}{llcccccccccc}
\multicolumn{12}{l}{\textit{Panel A: analytic standard errors
($2{,}000$ replications)}}\\
\toprule
&&&& \multicolumn{2}{c}{full} & \multicolumn{2}{c}{no cross}
& \multicolumn{2}{c}{conv.} & MD & \\
\cmidrule(lr){5-6}\cmidrule(lr){7-8}\cmidrule(lr){9-10}
Configuration & $n$ & bias & sd & se & cov & se & cov & se & cov & cov
& $J$ \\
\midrule
$q=0.25$, $T=2$ & 250 & -0.001 & 0.106 & 0.108 & 0.946 & 0.169 & 0.996 & 0.170 & 0.998 & 0.935 & 0.059 \\
 & 500 & -0.000 & 0.075 & 0.076 & 0.952 & 0.119 & 0.998 & 0.119 & 0.998 & 0.945 & 0.058 \\
 & 1000 & 0.000 & 0.053 & 0.053 & 0.952 & 0.084 & 0.999 & 0.083 & 0.999 & 0.949 & 0.049 \\
\addlinespace
$q=0.25$, $T=4$ & 250 & 0.001 & 0.064 & 0.066 & 0.954 & 0.097 & 0.997 & 0.097 & 0.997 & 0.909 & 0.167 \\
 & 500 & -0.001 & 0.047 & 0.046 & 0.942 & 0.068 & 0.994 & 0.068 & 0.994 & 0.929 & 0.091 \\
 & 1000 & -0.001 & 0.033 & 0.033 & 0.949 & 0.048 & 0.997 & 0.048 & 0.996 & 0.938 & 0.068 \\
\addlinespace
$q=0.5$, $T=2$ & 250 & 0.002 & 0.097 & 0.102 & 0.951 & 0.166 & 0.998 & 0.166 & 0.998 & 0.954 & 0.045 \\
 & 500 & -0.001 & 0.070 & 0.072 & 0.952 & 0.116 & 0.999 & 0.116 & 0.999 & 0.952 & 0.050 \\
 & 1000 & 0.001 & 0.050 & 0.051 & 0.959 & 0.082 & 1.000 & 0.082 & 1.000 & 0.959 & 0.040 \\
\addlinespace
$q=0.5$, $T=4$ & 250 & 0.001 & 0.061 & 0.063 & 0.954 & 0.095 & 0.998 & 0.095 & 0.998 & 0.906 & 0.106 \\
 & 500 & 0.001 & 0.043 & 0.044 & 0.952 & 0.067 & 0.998 & 0.067 & 0.998 & 0.933 & 0.066 \\
 & 1000 & -0.001 & 0.030 & 0.031 & 0.954 & 0.047 & 0.999 & 0.047 & 0.999 & 0.948 & 0.051 \\
\addlinespace
$q=0.75$, $T=2$ & 250 & -0.001 & 0.122 & 0.118 & 0.936 & 0.189 & 0.995 & 0.191 & 0.997 & 0.933 & 0.053 \\
 & 500 & -0.003 & 0.086 & 0.085 & 0.947 & 0.135 & 0.996 & 0.135 & 0.998 & 0.943 & 0.043 \\
 & 1000 & 0.002 & 0.062 & 0.060 & 0.948 & 0.095 & 0.997 & 0.096 & 0.997 & 0.946 & 0.048 \\
\addlinespace
$q=0.75$, $T=4$ & 250 & 0.002 & 0.076 & 0.072 & 0.936 & 0.107 & 0.994 & 0.109 & 0.993 & 0.876 & 0.193 \\
 & 500 & -0.000 & 0.054 & 0.052 & 0.943 & 0.077 & 0.993 & 0.077 & 0.993 & 0.916 & 0.115 \\
 & 1000 & 0.000 & 0.037 & 0.037 & 0.947 & 0.055 & 0.997 & 0.055 & 0.997 & 0.934 & 0.078 \\
\addlinespace
\bottomrule
\end{tabular}

\medskip
\begin{tabular}{llcccccccccc}
\multicolumn{12}{l}{\textit{Panel B: pairs cluster bootstrap}}\\
\toprule
&& \multicolumn{5}{c}{$\widehat\beta_{EW}$}
& \multicolumn{5}{c}{$\widehat\beta_{MD}$} \\
\cmidrule(lr){3-7}\cmidrule(lr){8-12}
&& sd & se & cov & se$_{boot}$ & cov$_{boot}$
& sd & se & cov & se$_{boot}$ & cov$_{boot}$ \\
\midrule
$q=0.25$ & 500 & 0.076 & 0.076 & 0.936 & 0.082 & 0.968 & 0.077 & 0.075 & 0.937 & 0.082 & 0.965 \\
$q=0.5$ & 500 & 0.068 & 0.072 & 0.960 & 0.077 & 0.970 & 0.068 & 0.072 & 0.959 & 0.077 & 0.969 \\
$q=0.75$ & 500 & 0.087 & 0.085 & 0.942 & 0.096 & 0.973 & 0.088 & 0.084 & 0.937 & 0.095 & 0.968 \\
\bottomrule
\end{tabular}}
\end{table}

The point estimates confirm the theory: $\widehat\beta_{EW}$ and
$\widehat\beta_{MD}$ are unbiased at all three quantiles despite the
misspecified first step (largest absolute mean bias $0.003$). The
inference comparison is stark. The full sandwich tracks the Monte
Carlo standard deviation closely---average standard errors are within
$2\%$ of the empirical standard deviations at the median and within
$5\%$ in the least favorable outer-quantile configuration with
$n=250$---and delivers coverage between $0.936$ and $0.959$ in all
eighteen configurations. Both shortcut variance estimators fail in the
same way: the sampling errors of the period-$s$ and period-$t$
projection coefficients are strongly positively correlated because they
are computed from the same individuals, and the variance of the contrast
$\widehat b_{st}$ subtracts twice that covariance. Zeroing the
cross-period blocks discards the subtraction and overstates the variance
by $50$--$70\%$ in this design, so coverage rises to
$0.993$--$0.999$. Confidence intervals built from conventional
quantile-regression output are thus badly miscalibrated
even when they err on the conservative side; nothing guarantees the
direction of the distortion in general, since it is governed by the sign
pattern of the neglected blocks. The overidentification test retains
approximately correct size under misspecified projections for $T=2$
(rejection rates $0.040$--$0.059$), with the same small-$n$
overrejection at $T=4$ noted before ($0.11$--$0.19$ at $n=250$, falling
to $0.05$--$0.08$ at $n=1000$). Panel B shows that the pairs cluster
bootstrap, with the minimum-distance weight held fixed at the
full-sample estimate across draws, is a serviceable alternative to the
analytic sandwich, with a mild conservative tilt at this sample size:
across the three quantiles the mean bootstrap standard error exceeds the
Monte Carlo standard deviation by about $9$--$13\%$ (for example $0.082$
against an empirical $0.076$ at $q=0.25$), so bootstrap coverage lands
at $0.965$--$0.973$ where the analytic sandwich delivers
$0.936$--$0.960$. This slight upward bias of the nonparametric
bootstrap for quantile-regression variability in moderate samples is
well documented; the analytic robust intervals are the closer to
nominal here, and the bootstrap is best viewed as the conservative
fallback when programming the stacked sandwich is impractical.

Bandwidth sensitivity for the kernel-estimated Jacobian is examined in
Table~\ref{mc:tab:bw}, which multiplies the rule-of-thumb $h$ by $0.5$,
$1$, and $2$ in one Design 1 and one Design 3 configuration.
Coverage of the equally weighted intervals moves only from $0.938$ to
$0.975$ (Design 1 configuration) and $0.941$ to $0.971$ (Design 3
configuration) as the bandwidth varies over a fourfold range, with the
rule-of-thumb value closest to $0.95$ in both cases. The
overidentification statistic is more sensitive in the direction one
would expect: halving the bandwidth makes the kernel Jacobian noisier
and inflates the rejection rate of a true null to $0.14$ in the Design 1
configuration ($T=4$), while doubling it is conservative ($0.02$);
in the $T=2$ Design 3 configuration the rejection rate stays within
$0.050$--$0.069$ throughout. The rule-of-thumb bandwidth is thus
adequate for estimation and intervals, and for testing purposes erring
toward larger rather than smaller bandwidths is the safer deviation.

\subsection{Design 4: heavy tails and the comparison with least
squares}\label{mc:d4}

The fourth design quantifies the efficiency trade-off between the
quantile-based procedure and least-squares fixed effects. The DGP is the
Design 1 baseline with $\rho_A=0.75$, $\rho_V=0.5$, with four error
distributions for the innovations driving $V_{it}$: (i) the Gaussian
AR(1); (ii) a contaminated version in which, with probability $0.05$, an
individual's entire error path is scaled by $8$; (iii) i.i.d.\
Student-$t_3$ errors scaled to unit variance; and (iv) i.i.d.\
Student-$t_{1.5}$ errors, which have infinite variance, scaled to match
the Gaussian interquartile range. We compare the median contrast
estimators $\widehat\beta_{EW}(0.5)$ and $\widehat\beta_{MD}(0.5)$, the
composite minimum-distance estimator combining
$q\in\{0.25,0.5,0.75\}$, a least-squares analogue of the projection
procedure (period-specific least-squares projections on the history,
combined by minimum distance), and the within (fixed-effects) estimator
with cluster-robust standard errors. Table~\ref{mc:tab:d4} reports
RMSE, coverage, and the $0.95$ quantile of the absolute error, which
summarizes tail risk of the point estimate.

\begin{table}[t]
\centering
\caption{Design 4: error distributions and the comparison with least
squares ($\rho_A=0.75$, $\rho_V=0.5$, $q=0.5$; composite combines
$q\in\{0.25,0.5,0.75\}$; $2{,}000$ replications). MD(0.5) is the
median minimum-distance estimator, CMD the composite cross-quantile
estimator, LS-MD the least-squares projection analogue, and FE the
within estimator with cluster-robust standard errors.
$\mathrm{Q}_{.95}|err|$ is the $0.95$ quantile of
$|\widehat\beta-\beta_0|$.}
\label{mc:tab:d4}
\setlength{\tabcolsep}{3.2pt}
{\small
\begin{tabular}{lllcccccccccc}
\toprule
&&& \multicolumn{2}{c}{MD(0.5)} & \multicolumn{2}{c}{CMD}
& \multicolumn{2}{c}{LS-MD} & \multicolumn{2}{c}{FE}
& \multicolumn{2}{c}{$\mathrm{Q}_{.95}|err|$} \\
\cmidrule(lr){4-5}\cmidrule(lr){6-7}\cmidrule(lr){8-9}
\cmidrule(lr){10-11}\cmidrule(lr){12-13}
Errors & $T$ & $n$ & rmse & cov & rmse & cov & rmse & cov & rmse & cov
& MD(0.5) & FE \\
\midrule
Gaussian & 2 & 500 & 0.070 & 0.950 & 0.055 & 0.939 & 0.045 & 0.947 & 0.045 & 0.945 & 0.140 & 0.088 \\
 & 2 & 1000 & 0.050 & 0.949 & 0.038 & 0.947 & 0.032 & 0.947 & 0.032 & 0.949 & 0.097 & 0.062 \\
 & 4 & 500 & 0.043 & 0.945 & 0.034 & 0.909 & 0.027 & 0.940 & 0.029 & 0.952 & 0.084 & 0.057 \\
 & 4 & 1000 & 0.030 & 0.947 & 0.024 & 0.932 & 0.019 & 0.943 & 0.021 & 0.952 & 0.060 & 0.040 \\
\addlinespace
Contaminated & 2 & 500 & 0.075 & 0.944 & 0.058 & 0.946 & 0.086 & 0.953 & 0.088 & 0.954 & 0.149 & 0.174 \\
 & 2 & 1000 & 0.053 & 0.955 & 0.040 & 0.950 & 0.063 & 0.959 & 0.063 & 0.962 & 0.103 & 0.124 \\
 & 4 & 500 & 0.045 & 0.944 & 0.037 & 0.912 & 0.045 & 0.946 & 0.059 & 0.954 & 0.089 & 0.115 \\
 & 4 & 1000 & 0.031 & 0.951 & 0.025 & 0.933 & 0.035 & 0.936 & 0.042 & 0.952 & 0.061 & 0.081 \\
\addlinespace
$t_3$ & 2 & 500 & 0.065 & 0.955 & 0.054 & 0.943 & 0.061 & 0.946 & 0.063 & 0.947 & 0.126 & 0.124 \\
 & 4 & 500 & 0.038 & 0.944 & 0.032 & 0.919 & 0.034 & 0.932 & 0.036 & 0.944 & 0.073 & 0.072 \\
\addlinespace
$t_{1.5}$ & 2 & 500 & 0.090 & 0.963 & 0.084 & 0.944 & 0.386 & 0.959 & 2.018 & 0.964 & 0.178 & 0.878 \\
 & 4 & 500 & 0.052 & 0.954 & 0.049 & 0.927 & 0.109 & 0.947 & 0.428 & 0.960 & 0.102 & 0.606 \\
\addlinespace
\bottomrule
\end{tabular}}
\end{table}

Under Gaussian errors the ordering is as expected: least squares is
most precise, and the median-based estimator pays an efficiency cost of
just over $50\%$ in RMSE ($0.070$ versus $0.045$ at $T=2$, $n=500$),
which the composite estimator cuts to about $20\%$ ($0.055$). Under
contamination the ranking reverses---the composite estimator's RMSE of
$0.058$ beats fixed effects' $0.088$ by a third---and under $t_3$ errors
the composite estimator is modestly better than fixed effects while the
median alone is comparable to it. The infinite-variance
$t_{1.5}$ case makes the robustness point dramatically: the
fixed-effects estimator has RMSE $2.02$ and a $0.95$ absolute-error
quantile of $0.88$ at $T=2$, $n=500$, while the median and composite
estimators keep RMSEs of $0.090$ and $0.084$ with tail quantiles below
$0.18$; the least-squares projection analogue fails similarly
($\mathrm{RMSE}=0.39$). Notably, cluster-robust fixed-effects intervals
still ``cover'' in the $t_{1.5}$ design, but only because the standard
errors explode along with the estimator; the interval width, not the
coverage rate, reveals the breakdown. Coverage for the quantile-based
procedures stays in the $0.94$--$0.96$ range throughout, with the
composite estimator's joint test and intervals showing the usual
degradation at $T=4$, $n=500$, where the stacked weight matrix has
dimension $36$ (composite coverage $0.91$--$0.93$; joint-statistic
rejection of a true null $0.16$--$0.19$, falling by half at
$n=1000$).

\subsection{Design 5: detecting violations of stationarity}\label{mc:d5}

\subsubsection*{5A. A location shift}

The first violation adds a period-two location shift tied to the
period-one regressor: $V_{i2}$ is replaced by $V_{i2}+\kappa X_{i1}$ in
the Design 1 baseline ($T=2$, $\rho_A=0.75$, $\rho_V=0.5$, $q=0.5$),
$\kappa\in\{0,0.1,0.25,0.5\}$. Under this violation the population
value of the $(1,2)$ contrast is $\beta_0-\kappa$ while the $(2,1)$
contrast still equals $\beta_0$, so the overidentifying restriction
fails and the minimum-distance estimand is pulled toward the average.
The simulations reproduce this arithmetic exactly---mean
$\widehat b_{12}=1-\kappa$ and mean $\widehat b_{21}=1.00$ to three
decimals at every $\kappa$ and $n$ (Table~\ref{mc:tab:d5a})---so
divergence between contrasts is informative about the failure, and the
$J_n(0.5)$ statistic detects it: Figure~\ref{mc:fig:d5a} plots rejection
frequencies against $\kappa$. Size at $\kappa=0$ is $0.044$--$0.054$;
power reaches $0.41$ at the modest violation $\kappa=0.1$ for $n=1000$,
reaches $0.85$ at $n=500$ and $0.99$ at $n=1000$ for $\kappa=0.25$, and
is essentially one at $\kappa=0.5$ for all $n$.

\begin{figure}[t]
\centering
\includegraphics[width=0.5\textwidth]{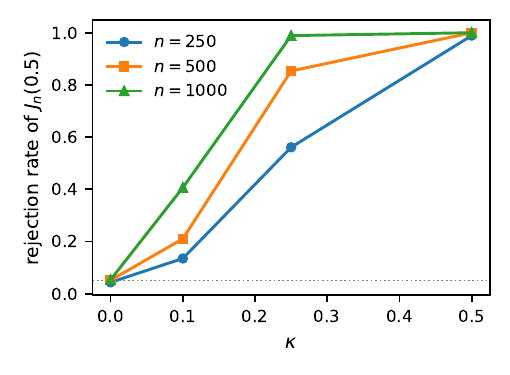}
\caption{Design 5A. Rejection frequencies of the projection-restriction
statistic $J_n(0.5)$ at the $5\%$ level (dotted line) against the
location-shift parameter $\kappa$, by sample size. $T=2$; $2{,}000$
replications per point.}
\label{mc:fig:d5a}
\end{figure}

\subsubsection*{5B. Quantile-varying slopes and the cross-quantile test}

The second violation is the one singled out by
Theorem~\ref{thm:incompatibility}: a random-coefficient design in which
every single-quantile restriction holds exactly while the constant-slope
model is false. With $T=2$, let $X_{it}=\exp\{0.5(C_i+Z_{it})\}$,
$A_i=0.75\,\overline X_i$, $U_{it}\sim\mathrm{Uniform}(0,1)$ i.i.d., and
\begin{equation}
Y_{it}=A_i+X_{it}\bigl\{\beta_0+\kappa\,(U_{it}-\tfrac12)\bigr\},
\qquad
\beta(q)=\beta_0+\kappa\,(q-\tfrac12).
\label{mc:eq:d5b}
\end{equation}
For each fixed $q$ the diagonal-minus-off-diagonal restrictions hold
exactly with value $\beta(q)$, so $J_n(q)$ should remain near size for
every $q$; only the cross-quantile restriction
$\beta(q_1)=\dots=\beta(q_L)$ is false when $\kappa\neq0$. We use
$q\in\{0.1,0.25,0.5,0.75,0.9\}$, $\kappa\in\{0,0.25,0.5,1\}$, and
$n\in\{500,1000\}$, and compare the per-quantile statistics $J_n(q)$
(one degree of freedom each) with the joint cross-quantile statistic
($rL-1=9$ degrees of freedom). Figure~\ref{mc:fig:d5b} contains the two
displays suggested by the design: estimated quantile-specific slopes
against the population line, and rejection frequencies.

\begin{figure}[t]
\centering
\includegraphics[width=\textwidth]{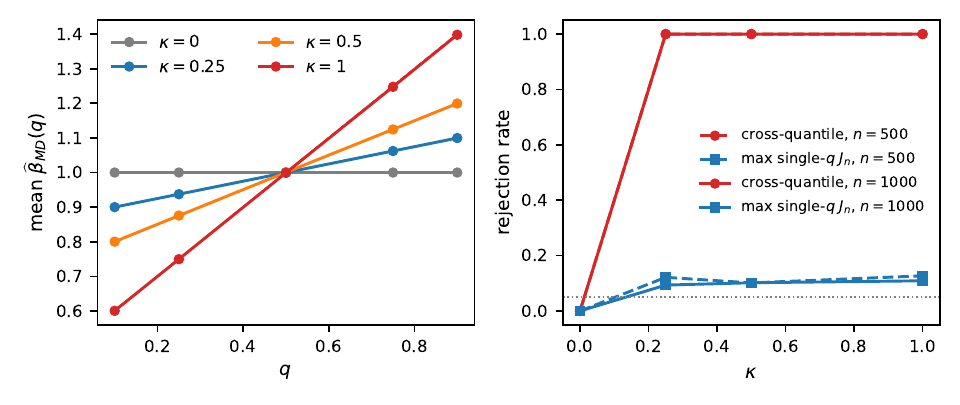}
\caption{Design 5B. Left: mean of $\widehat\beta_{MD}(q)$ across
replications (markers) and the population line
$\beta_0+\kappa(q-\tfrac12)$ (dashed) at $n=1000$. Right: rejection
frequencies of the joint cross-quantile statistic (red) and of the most
frequently rejecting single-quantile statistic $J_n(q)$ (blue) against
$\kappa$; solid lines $n=1000$, dashed $n=500$; the dotted line marks
$0.05$. At $\kappa=0$ the design is noiseless and all statistics are
degenerate at zero.}
\label{mc:fig:d5b}
\end{figure}

Both predictions are borne out with unusual precision. The estimated
quantile-specific slopes lie on the population line: at $n=1000$ every
mean $\widehat\beta_{MD}(q)$ is within $0.003$ of
$\beta_0+\kappa(q-\tfrac12)$, so the procedure applied quantile by
quantile consistently recovers the entire slope process. Each
single-quantile test stays near its nominal level at every
$\kappa$---rejection rates range over $0.06$--$0.13$, drifting toward
size as $n$ grows---because the restrictions it examines are true. The
joint cross-quantile statistic, in contrast, rejects with probability
one at every $\kappa\ge0.25$ and both sample sizes. This is the
practical content of the incompatibility theorem: single-quantile
statistics test the quantile-projection restrictions, not the equality
of slopes across quantiles, and a researcher who wants power against
quantile-varying coefficients must compare estimates across quantiles.
(At $\kappa=0$ the outcome is an exact linear function of the regressors
and the design is degenerate: all estimators equal $\beta_0$ and all
statistics equal zero identically, so the $\kappa=0$ column speaks to
neither size nor power.)

\subsubsection*{5C. Pure scale nonstationarity}

The last violation rescales the period-two error,
$V_{i2}=(1+\kappa|X_{i1}|)G_{i2}$ with $G_{it}$ i.i.d.\ standard normal
and the baseline correlated effect, $\kappa\in\{0,0.5,1\}$, $n=500$,
$T=2$. Table~\ref{mc:tab:d5c} reports rejection frequencies of
$J_n(q)$ at $q\in\{0.25,0.5,0.75\}$ and of the joint statistic across
the three quantiles. As anticipated, the median-based test has
essentially no power (rejection rate $0.11$ at $\kappa=1$), because a
symmetric rescaling leaves conditional medians unchanged. More
interesting is that the tail statistics also have little power
(about $0.10$ at $\kappa=1$; joint statistic $0.13$), and the design
explains why: the induced shift in the period-two quantile function is
proportional to $|X_{i1}|$, which is uncorrelated with $X_{i1}$ under
the symmetric regressor design, so the {\it linear} quantile projections
that the statistic examines are almost unaffected---the violation is
absorbed mainly by the intercept. This is a useful caution that
sharpens the interpretation given in Section 5B: the tests have power
against nonstationarity only insofar as it moves the linear
quantile-projection coefficients, and scale shifts orthogonal to the
included regressors can evade quantile-specific and cross-quantile tests
alike. Detection of such violations requires enriching the projection
set (for example with $|X_{i1}|$ itself), which the framework
accommodates directly.

\begin{table}[t]
\centering
\caption{Designs 5A and 5C. Left: location shift ($T=2$, $q=0.5$);
reported are the means across replications of the two contrasts and of
$\widehat\beta_{MD}$, whose population values under the violation are
$\beta_0-\kappa$, $\beta_0$, and (approximately) $\beta_0-\kappa/2$, and
the rejection rate of $J_n(0.5)$. Right: pure scale nonstationarity
($T=2$, $n=500$); rejection rates of $J_n(q)$ and of the joint
statistic across $q\in\{0.25,0.5,0.75\}$.}
\label{mc:tab:d5a}
\setlength{\tabcolsep}{4pt}
{\small
\begin{tabular}[t]{llcccc}
\multicolumn{6}{l}{\textit{Design 5A}}\\
\toprule
& $n$ & $\bar b_{12}$ & $\bar b_{21}$ & $\bar\beta_{MD}$ & $J$ rej. \\
\midrule
$\kappa=0$ & 250 & 1.005 & 1.004 & 1.005 & 0.043 \\
 & 500 & 0.999 & 1.000 & 0.999 & 0.051 \\
 & 1000 & 1.001 & 1.000 & 1.000 & 0.053 \\
\addlinespace
$\kappa=0.1$ & 250 & 0.901 & 1.001 & 0.951 & 0.135 \\
 & 500 & 0.903 & 1.000 & 0.951 & 0.210 \\
 & 1000 & 0.899 & 1.001 & 0.950 & 0.406 \\
\addlinespace
$\kappa=0.25$ & 250 & 0.754 & 1.004 & 0.880 & 0.561 \\
 & 500 & 0.749 & 0.998 & 0.873 & 0.854 \\
 & 1000 & 0.752 & 0.999 & 0.876 & 0.989 \\
\addlinespace
$\kappa=0.5$ & 250 & 0.497 & 1.001 & 0.748 & 0.989 \\
 & 500 & 0.498 & 0.997 & 0.748 & 1.000 \\
 & 1000 & 0.500 & 0.999 & 0.750 & 1.000 \\
\addlinespace
\bottomrule
\end{tabular}\hspace{1.5em}
\begin{tabular}[t]{lcccc}
\multicolumn{5}{l}{\textit{Design 5C}}\\
\toprule
$\kappa$ & $J(.25)$ & $J(.5)$ & $J(.75)$ & joint \\
\midrule
0 & 0.043 & 0.044 & 0.046 & 0.057 \\
0.5 & 0.063 & 0.064 & 0.065 & 0.070 \\
1 & 0.104 & 0.107 & 0.099 & 0.132 \\
\bottomrule
\end{tabular}}
\label{mc:tab:d5c}
\end{table}

\subsection{Toward longer panels}\label{mc:d6}

Although the theory is designed for short panels, users will ask how the
procedure behaves as $T$ grows and how it compares with large-$T$
two-step methods. We therefore run the Design 1 baseline
($\rho_A=0.75$, $\rho_V=0.5$, $n=500$, $q=0.5$) at $T=10$ and $T=20$,
where the number of contrasts is $r=90$ and $r=380$, and add the
two-step estimator of the large-$T$ literature that first estimates the
fixed effects from a within regression and then runs a pooled quantile
regression of $Y_{it}-\widehat A_i$ on $X_{it}$. Results are in
Table~\ref{mc:tab:d6}.

\begin{table}[t]
\centering
\caption{Longer panels ($\rho_A=0.75$, $\rho_V=0.5$, $n=500$, $q=0.5$,
$2{,}000$ replications). ``Canay'' is the large-$T$ two-step estimator
(within slope, estimated fixed effects removed, pooled median
regression); FE is the within estimator. $r=T(T-1)$ is the number of
contrasts.}
\label{mc:tab:d6}
\setlength{\tabcolsep}{4pt}
{\small
\begin{tabular}{lcccccccccc}
\toprule
& \multicolumn{3}{c}{EW} & \multicolumn{3}{c}{MD} & FE & Canay & \\
\cmidrule(lr){2-4}\cmidrule(lr){5-7}
$T$ & bias & rmse & cov & bias & rmse & cov & rmse & rmse & $J$ \\
\midrule
10 & -0.001 & 0.027 & 0.936 & -0.001 & 0.029 & 0.854 & 0.020 & 0.022 & 0.643 \\
20 & 0.000 & 0.019 & 0.939 & -0.000 & 0.041 & 0.259 & 0.014 & 0.016 & 1.000 \\
\bottomrule
\end{tabular}}
\end{table}

The comparison delineates the method's range of application. The
equally weighted estimator remains unbiased with reliable inference as
the panel lengthens---coverage is $0.936$--$0.939$ at both $T$---and its
precision continues to improve with $T$ (RMSE $0.027$ at $T=10$,
$0.019$ at $T=20$). The feasible minimum-distance estimator, in
contrast, breaks down exactly where the asymptotic approximation for the
estimated weight matrix must fail: with $r=90$ contrasts and $n=500$ its
coverage falls to $0.85$ and $J_n$ rejects a true null $64\%$ of the
time, and with $r=380$---a covariance matrix of nearly the same
dimension as the sample size---its Monte Carlo standard deviation is
double that of equal weighting, coverage collapses to $0.26$, and $J_n$
rejects always. Notably the $10^{12}$ condition-number safeguard is
never triggered even at $T=20$: the failure is statistical, not
numerical, and would not be flagged by a rank check. The within
estimator is the most precise in this Gaussian design, and the large-$T$
two-step estimator tracks it closely (RMSE $0.022$ and $0.016$),
outperforming equal weighting here because its incidental-parameter
bias is negligible under symmetric errors at these values of $T$; the
short-panel guarantees of the present method are, of course, exactly
what that estimator lacks at $T=2$ or $4$. The practical guidance is
clear: as $T$ grows, use the equally weighted estimator (or a
minimum-distance estimator on a reduced set of contrasts, such as
adjacent pairs) and reserve full efficient weighting and the
overidentification test for panels with $r=T(T-1)$ small relative to
$n$.

\subsection{Summary}

Across designs, four practical messages emerge. The contrast-averaging
estimators are unbiased and reliably covered by the robust intervals in
every correctly specified configuration, with the equally weighted
estimator the more dependable choice when the number of contrasts is
large relative to $n$ and the efficiency loss from equal weighting
bounded by single-digit percentages in all our Gaussian designs.
Conventional quantile-regression standard errors that ignore
cross-period dependence can be off by more than $50\%$ even when the
point estimator is consistent; the stacked cluster sandwich or the
(mildly conservative) pairs bootstrap should be regarded as integral to
the method. The efficiency
cost of the median-based procedure under Gaussian errors is roughly
offset by combining a few quantiles, and the protection purchased is
large: under infinite-variance errors the fixed-effects estimator's RMSE
is more than twenty times that of the composite estimator. Finally, the
specification statistics behave as the identification theory predicts:
they detect violations that move the linear quantile projections, the
cross-quantile statistic---and only the cross-quantile
statistic---detects quantile-varying slopes, and violations orthogonal
to the projection space require augmenting the projections to be
detectable.


\begin{table}[p]
\centering
\caption{Design 1, complete grid ($\sigma_X=1$, $q=0.5$, $2{,}000$
replications). ``vr'' columns are the ratios of Monte Carlo variances
$\mathrm{Var}(\widehat\beta_{MD})/\mathrm{Var}(\widehat\beta_{EW})$ and
$\mathrm{Var}(\widehat\beta_{OMD})/\mathrm{Var}(\widehat\beta_{EW})$.}
\label{mc:tab:d1full}
\setlength{\tabcolsep}{1.7pt}
{\scriptsize
\begin{tabular}{llccccccccccccccc}
\toprule
&& \multicolumn{3}{c}{SC} & \multicolumn{3}{c}{EW}
& \multicolumn{3}{c}{MD} & \multicolumn{2}{c}{OMD} && \\
\cmidrule(lr){3-5}\cmidrule(lr){6-8}\cmidrule(lr){9-11}
\cmidrule(lr){12-13}
Configuration & $n$ & bias & rmse & cov & bias & rmse & cov
& mbias & rmse & cov & rmse & cov & vr$_{MD}$ & vr$_{OMD}$ & $J$ \\
\midrule
$\rho_A=0$, $\rho_V=0$, $T=2$ & 250 & -0.005 & 0.151 & 0.957 & -0.002 & 0.132 & 0.950 & -0.002 & 0.132 & 0.944 & 0.132 & 0.945 & 1.002 & 1.000 & 0.061 \\
 & 500 & 0.000 & 0.102 & 0.960 & 0.001 & 0.090 & 0.960 & -0.001 & 0.090 & 0.958 & 0.090 & 0.955 & 1.003 & 1.000 & 0.045 \\
 & 1000 & -0.002 & 0.075 & 0.954 & -0.003 & 0.065 & 0.955 & -0.003 & 0.065 & 0.954 & 0.065 & 0.952 & 0.996 & 1.000 & 0.045 \\
\addlinespace
$\rho_A=0$, $\rho_V=0.7$, $T=2$ & 250 & -0.003 & 0.109 & 0.949 & -0.002 & 0.095 & 0.948 & 0.001 & 0.094 & 0.946 & 0.095 & 0.942 & 0.994 & 1.000 & 0.049 \\
 & 500 & 0.001 & 0.079 & 0.945 & 0.001 & 0.069 & 0.946 & -0.000 & 0.069 & 0.941 & 0.069 & 0.942 & 1.011 & 1.000 & 0.052 \\
 & 1000 & 0.001 & 0.055 & 0.944 & 0.001 & 0.047 & 0.956 & 0.001 & 0.047 & 0.953 & 0.047 & 0.950 & 1.008 & 1.000 & 0.050 \\
\addlinespace
$\rho_A=0$, $\rho_V=0$, $T=4$ & 250 & 0.004 & 0.162 & 0.957 & 0.003 & 0.074 & 0.952 & 0.002 & 0.078 & 0.926 & 0.074 & 0.951 & 1.106 & 1.000 & 0.087 \\
 & 500 & -0.001 & 0.114 & 0.958 & -0.000 & 0.052 & 0.955 & -0.001 & 0.053 & 0.943 & 0.052 & 0.949 & 1.046 & 1.000 & 0.068 \\
 & 1000 & -0.000 & 0.080 & 0.962 & 0.000 & 0.038 & 0.948 & 0.000 & 0.039 & 0.941 & 0.038 & 0.943 & 1.022 & 1.000 & 0.050 \\
\addlinespace
$\rho_A=0$, $\rho_V=0.7$, $T=4$ & 250 & -0.002 & 0.120 & 0.943 & -0.001 & 0.062 & 0.947 & -0.002 & 0.063 & 0.902 & 0.059 & 0.943 & 1.042 & 0.915 & 0.112 \\
 & 500 & -0.003 & 0.085 & 0.954 & -0.002 & 0.042 & 0.955 & -0.003 & 0.043 & 0.923 & 0.041 & 0.944 & 1.024 & 0.950 & 0.081 \\
 & 1000 & -0.001 & 0.059 & 0.955 & 0.001 & 0.029 & 0.957 & 0.000 & 0.029 & 0.943 & 0.029 & 0.949 & 0.984 & 0.940 & 0.057 \\
\addlinespace
$\rho_A=0.75$, $\rho_V=0$, $T=2$ & 250 & 0.004 & 0.143 & 0.951 & 0.002 & 0.125 & 0.950 & 0.000 & 0.125 & 0.946 & 0.125 & 0.941 & 1.001 & 1.000 & 0.038 \\
 & 500 & -0.001 & 0.098 & 0.957 & -0.002 & 0.084 & 0.963 & -0.002 & 0.085 & 0.960 & 0.084 & 0.954 & 1.003 & 1.000 & 0.052 \\
 & 1000 & 0.002 & 0.070 & 0.957 & 0.002 & 0.060 & 0.955 & 0.002 & 0.060 & 0.957 & 0.060 & 0.946 & 1.001 & 1.000 & 0.046 \\
\addlinespace
$\rho_A=0.75$, $\rho_V=0.7$, $T=2$ & 250 & 0.002 & 0.099 & 0.948 & 0.000 & 0.087 & 0.949 & -0.001 & 0.088 & 0.945 & 0.087 & 0.946 & 1.010 & 1.000 & 0.045 \\
 & 500 & 0.000 & 0.072 & 0.946 & -0.000 & 0.062 & 0.954 & 0.001 & 0.062 & 0.954 & 0.062 & 0.952 & 1.004 & 1.000 & 0.046 \\
 & 1000 & 0.000 & 0.051 & 0.957 & -0.000 & 0.044 & 0.958 & 0.000 & 0.044 & 0.955 & 0.044 & 0.953 & 1.011 & 1.000 & 0.050 \\
\addlinespace
$\rho_A=0.75$, $\rho_V=0$, $T=4$ & 250 & -0.004 & 0.158 & 0.938 & -0.001 & 0.070 & 0.951 & 0.001 & 0.073 & 0.930 & 0.070 & 0.949 & 1.089 & 1.000 & 0.088 \\
 & 500 & 0.003 & 0.106 & 0.957 & 0.001 & 0.050 & 0.949 & -0.000 & 0.050 & 0.944 & 0.050 & 0.949 & 1.009 & 1.000 & 0.060 \\
 & 1000 & -0.002 & 0.077 & 0.954 & -0.001 & 0.035 & 0.949 & -0.001 & 0.035 & 0.941 & 0.035 & 0.946 & 1.018 & 1.000 & 0.052 \\
\addlinespace
$\rho_A=0.75$, $\rho_V=0.7$, $T=4$ & 250 & -0.002 & 0.110 & 0.947 & 0.000 & 0.057 & 0.945 & 0.003 & 0.058 & 0.913 & 0.054 & 0.944 & 1.037 & 0.910 & 0.099 \\
 & 500 & -0.001 & 0.077 & 0.955 & 0.001 & 0.039 & 0.956 & 0.001 & 0.038 & 0.937 & 0.038 & 0.949 & 0.962 & 0.921 & 0.065 \\
 & 1000 & 0.000 & 0.055 & 0.952 & 0.000 & 0.027 & 0.957 & 0.001 & 0.026 & 0.953 & 0.026 & 0.950 & 0.933 & 0.924 & 0.061 \\
\addlinespace
\bottomrule
\end{tabular}}
\end{table}

\begin{table}[p]
\centering
\caption{Design 2, complete results ($\rho_A=0.75$, $\rho_V=0.5$,
$q=0.5$, $2{,}000$ replications). The parenthetical value after
$\sigma_X$ is $\mathrm{Var}(X_{it}-X_{is})$. cond$(\widehat J)$ and
cond$(\widehat\Omega)$ are medians across replications of the condition
numbers of the estimated Jacobian and contrast covariance matrix;
``adj.'' is the fraction of replications in which the $10^{12}$
pseudoinverse threshold was triggered.}
\label{mc:tab:d2}
\setlength{\tabcolsep}{2.8pt}
{\footnotesize
\begin{tabular}{llcccccccccc}
\toprule
Configuration & $n$ & \multicolumn{1}{c}{EW} & \multicolumn{1}{c}{MD}
& \multicolumn{1}{c}{OMD} & MD & $J$ & vr$_{MD}$ & vr$_{OMD}$
& cond$(\widehat J)$ & cond$(\widehat\Omega)$ & adj. \\
&& rmse & rmse & rmse & cov & rej. &&&&& \\
\midrule
$T=2$, $\sigma_X=1$ (1.000) & 500 & 0.071 & 0.072 & 0.071 & 0.956 & 0.049 & 1.005 & 1.000 & 3.1 & 3.1 & 0.000 \\
 & 1000 & 0.052 & 0.052 & 0.052 & 0.950 & 0.040 & 1.005 & 1.000 & 3.0 & 3.0 & 0.000 \\
\addlinespace
$T=2$, $\sigma_X=0.5$ (0.400) & 500 & 0.113 & 0.114 & 0.113 & 0.947 & 0.046 & 1.012 & 1.000 & 9.1 & 9.1 & 0.000 \\
 & 1000 & 0.077 & 0.077 & 0.077 & 0.954 & 0.051 & 1.005 & 1.000 & 9.1 & 9.0 & 0.000 \\
\addlinespace
$T=2$, $\sigma_X=0.25$ (0.118) & 500 & 0.202 & 0.203 & 0.202 & 0.957 & 0.051 & 1.009 & 1.000 & 33.3 & 33.4 & 0.000 \\
 & 1000 & 0.143 & 0.143 & 0.143 & 0.953 & 0.043 & 1.003 & 1.000 & 33.1 & 33.1 & 0.000 \\
\addlinespace
$T=4$, $\sigma_X=1$ (1.000) & 500 & 0.043 & 0.043 & 0.042 & 0.951 & 0.066 & 0.997 & 0.953 & 5.7 & 25.9 & 0.000 \\
 & 1000 & 0.030 & 0.030 & 0.030 & 0.955 & 0.053 & 0.974 & 0.947 & 5.5 & 23.6 & 0.000 \\
\addlinespace
$T=4$, $\sigma_X=0.5$ (0.400) & 500 & 0.069 & 0.069 & 0.067 & 0.935 & 0.068 & 0.996 & 0.948 & 19.3 & 83.8 & 0.000 \\
 & 1000 & 0.049 & 0.048 & 0.048 & 0.944 & 0.055 & 0.983 & 0.957 & 18.6 & 77.7 & 0.000 \\
\addlinespace
$T=4$, $\sigma_X=0.25$ (0.118) & 500 & 0.126 & 0.126 & 0.123 & 0.944 & 0.071 & 1.003 & 0.956 & 73.7 & 316.5 & 0.000 \\
 & 1000 & 0.091 & 0.090 & 0.089 & 0.940 & 0.057 & 0.974 & 0.948 & 71.1 & 294.9 & 0.000 \\
\addlinespace
$T=6$, $\sigma_X=1$ (1.000) & 500 & 0.035 & 0.036 & 0.034 & 0.921 & 0.129 & 1.028 & 0.917 & 8.5 & 78.3 & 0.000 \\
 & 1000 & 0.024 & 0.024 & 0.023 & 0.942 & 0.070 & 0.969 & 0.914 & 8.0 & 66.4 & 0.000 \\
\addlinespace
$T=6$, $\sigma_X=0.5$ (0.400) & 500 & 0.055 & 0.055 & 0.053 & 0.927 & 0.120 & 1.000 & 0.915 & 30.0 & 270.3 & 0.000 \\
 & 1000 & 0.038 & 0.038 & 0.037 & 0.938 & 0.086 & 0.999 & 0.933 & 28.5 & 230.6 & 0.000 \\
\addlinespace
$T=6$, $\sigma_X=0.25$ (0.118) & 500 & 0.102 & 0.103 & 0.098 & 0.927 & 0.115 & 1.030 & 0.932 & 116.6 & $1.0\times 10^{3}$ & 0.000 \\
 & 1000 & 0.071 & 0.070 & 0.068 & 0.940 & 0.067 & 1.000 & 0.931 & 110.4 & 890.4 & 0.000 \\
\addlinespace
\bottomrule
\end{tabular}}
\end{table}

\begin{table}[p]
\centering
\caption{Bandwidth sensitivity of the kernel-estimated Jacobian
($2{,}000$ replications). Entries are coverage of nominal $95\%$
intervals (EW/MD) and rejection rates of $J_n(q)$ when the
rule-of-thumb bandwidth $h$ is multiplied by $0.5$, $1$, and $2$.}
\label{mc:tab:bw}
\setlength{\tabcolsep}{4pt}
{\footnotesize
\begin{tabular}{lcccccc}
\toprule
& \multicolumn{3}{c}{coverage EW/MD} & \multicolumn{3}{c}{$J$ rej.} \\
\cmidrule(lr){2-4}\cmidrule(lr){5-7}
Configuration & $0.5h$ & $h$ & $2h$ & $0.5h$ & $h$ & $2h$ \\
\midrule
Design 1 ($\rho_A=0.75$, $\rho_V=0.7$, $T=4$, $q=0.5$) & 0.938/0.914 & 0.950/0.936 & 0.975/0.960 & 0.139 & 0.075 & 0.018 \\
Design 3 ($T=2$, $q=0.25$) & 0.941/0.934 & 0.956/0.952 & 0.971/0.968 & 0.069 & 0.062 & 0.050 \\
\bottomrule
\end{tabular}}
\end{table}

\section{Conclusion}
\label{sec:conclusion}

Conditional stationarity is often viewed as a relatively weak alternative to parametric assumptions on panel disturbances. In a quantile model it has a sharp implication: if quantile-specific residual distributions are stationary across periods and within-individual regressor changes have ordinary rank, the slope coefficient cannot vary with the quantile index. Stationarity therefore selects a common location effect rather than a family of quantile-specific slopes.

For the stationary location model, the same restriction provides a simple fixed-$T$ identification strategy. Period-specific cross-sectional quantile projections contain a common nuisance component, while the structural slope enters a different regressor block in each period. Diagonal-minus-off-diagonal coefficient contrasts eliminate the nuisance projection and identify the common slope for any $T\geq 2$. Pooling all contrasts by minimum distance yields a computationally simple estimator with standard $\sqrt n$ inference, arbitrary within-individual dependence, and no estimation of individual effects.

The approach also supplies specification diagnostics. Equality of the contrast estimates is testable even with two periods, and the restrictions can be strengthened by combining several quantile projections. These features make the method useful both as an estimator under the stationary location model and as a disciplined way to assess whether stationarity is compatible with the distributional heterogeneity in a given panel.

\appendix

\section{Proofs}
\label{app:proofs}

\subsection{Proof of Theorem \ref{thm:incompatibility} and Corollary \ref{cor:generic-failure}}

\begin{proof}[Proof of Theorem \ref{thm:incompatibility}]
Fix $t\neq s$ and $\tau,u\in\mathcal T$. Translation equivariance of generalized inverse quantiles and \eqref{eq:linear-quantile-model} give
\begin{align}
Q_u\{\varepsilon_{it}(\tau)\mid X_i,A_i\}
&=Q_u(Y_{it}\mid X_i,A_i)
  -c(\tau)-X_{it}'\beta(\tau)-\rho(\tau)A_i\\
&=c(u)-c(\tau)+X_{it}'\{\beta(u)-\beta(\tau)\}
  +\{\rho(u)-\rho(\tau)\}A_i.
\label{eq:proof-q-t}
\end{align}
Similarly,
\begin{equation}
Q_u\{\varepsilon_{is}(\tau)\mid X_i,A_i\}
 =c(u)-c(\tau)+X_{is}'\{\beta(u)-\beta(\tau)\}
  +\{\rho(u)-\rho(\tau)\}A_i.
\label{eq:proof-q-s}
\end{equation}
If the two conditional residual distributions are equal at a conditioning value, their generalized inverse quantile functions are equal. Subtracting \eqref{eq:proof-q-s} from \eqref{eq:proof-q-t} yields
\[
\Delta X_{its}'\{\beta(u)-\beta(\tau)\}=0.
\]
This proves part (i).

For part (ii), residual stationarity at $\tau$ makes the last equality hold almost surely for every $u\in\mathcal T$. Let $\gamma_u=\beta(u)-\beta(\tau)$. Then
\[
0=\E[(\Delta X_{its}'\gamma_u)^2]
 =\gamma_u'\E[\Delta X_{its}\Delta X_{its}']\gamma_u.
\]
Positive definiteness in Assumption \ref{ass:within-rank} implies $\gamma_u=0$. Since $u$ was arbitrary, $\beta(\cdot)$ is constant on $\mathcal T$.
\end{proof}

\begin{proof}[Proof of Corollary \ref{cor:generic-failure}]
Fix $\tau$. If $\beta(\cdot)$ is not constant, there exists $u_\tau\in\mathcal T$ such that
$\gamma_\tau=\beta(u_\tau)-\beta(\tau)\neq 0$. By part (i) of Theorem \ref{thm:incompatibility}, the event $\{(X_i,A_i)\in\mathcal S_{ts}(\tau)\}$ is contained in
\[
\{\Delta X_{its}'\gamma_\tau=0\}.
\]
The probability of this hyperplane event is zero by \eqref{eq:no-hyperplane}. The result follows.
\end{proof}

\subsection{Proofs for identification}

\begin{proof}[Proof of Lemma \ref{lem:composite-stationarity}]
For any real $r$ and conditioning value $X_i=x$,
\begin{align*}
\Pp(R_{it}\leq r\mid X_i=x)
&=\int \Pp(V_{it}\leq r-a\mid X_i=x,A_i=a)
 \,dF_{A\mid X}(a\mid x).
\end{align*}
Assumption \ref{ass:stationarity} makes the integrand identical for every $t$. The integral, and therefore the conditional distribution of $R_{it}$ given $X_i$, is common across periods.
\end{proof}

\begin{proof}[Proof of Proposition \ref{prop:projection-representation}]
By \eqref{eq:location-model} and \eqref{eq:selector},
\[
Y_{it}=X_i'E_t\beta_0+R_{it}.
\]
For any $\pi\in\R^{pT}$, write $\delta=\pi-E_t\beta_0$. Then
\begin{align*}
\E[\rho_q(Y_{it}-X_i'\pi)]
&=\E[\rho_q\{R_{it}-X_i'(\pi-E_t\beta_0)\}]\\
&=\E[\rho_q(R_{it}-X_i'\delta)].
\end{align*}
By Lemma \ref{lem:composite-stationarity}, this last objective is the same for every $t$, and by Assumption \ref{ass:unique-projection} it has the unique minimizer $\delta_0(q)$. Hence the unique minimizer in the original parameterization is
$\pi_{0t}(q)=E_t\beta_0+\delta_0(q)$. Premultiplying by $E_s'$ yields \eqref{eq:block-representation}.
\end{proof}

\begin{proof}[Proof of Theorem \ref{thm:identification}]
For $s\neq t$, Proposition \ref{prop:projection-representation} gives
\begin{align*}
E_s'\{\pi_{0s}(q)-\pi_{0t}(q)\}
&=E_s'(E_s-E_t)\beta_0\\
&=\beta_0,
\end{align*}
because $E_s'E_s=I_p$ and $E_s'E_t=0$. Thus any one contrast point identifies $\beta_0$ when $T\geq2$. When $T=1$, \eqref{eq:projection-representation} reduces to
$\pi_{01}(q)=\delta_0(q)+\beta_0$, which does not separate the two terms.
\end{proof}

\subsection{Matrix rank and equivalence}

\begin{lemma}[Rank properties]
\label{lem:rank-properties}
For $T\geq2$, $M=[G\ H]$ in \eqref{eq:G-H} has full column rank $p+pT$, $C$ has full row rank $pT(T-1)$, and
\begin{equation}
CH=0,
\qquad
CG=D.
\label{eq:annihilation-identities}
\end{equation}
Moreover, $\ker(C)=\operatorname{col}(H)$.
\end{lemma}

\begin{proof}
Suppose $G\gamma+H\eta=0$. Looking at the period-$t$ block gives
$E_t\gamma+\eta=0$ for every $t$. Hence $(E_t-E_s)\gamma=0$ for every pair. Since $E_t\gamma$ and $E_s\gamma$ place the same vector in disjoint blocks, this is possible only if $\gamma=0$, after which $\eta=0$. Thus $M$ has full column rank.

For a fixed regressor block $s$, the $T-1$ contrasts comparing equation $s$ to equations $t\neq s$ are linearly independent differences among the $T$ equation-specific coefficients on that block. Different values of $s$ involve disjoint coordinates. Therefore $C$ has rank $pT(T-1)$. The identities in \eqref{eq:annihilation-identities} follow directly from the definitions. Since $H$ has rank $pT$ and $CH=0$, $\operatorname{col}(H)\subseteq\ker(C)$. The two spaces have the same dimension because
\[
\dim\ker(C)=pT^2-pT(T-1)=pT.
\]
They are therefore equal.
\end{proof}

\begin{proof}[Proof of Proposition \ref{prop:equivalence}]
Write $\Sigma=\Sigma_{\Pi,q}$. Since $C$ has full row rank and $\ker(C)=\operatorname{col}(H)$, the generalized residual-maker identity gives
\begin{equation}
C'(C\Sigma C')^{-1}C
 =\Sigma^{-1}
 -\Sigma^{-1}H(H'\Sigma^{-1}H)^{-1}H'\Sigma^{-1}.
\label{eq:residual-maker-identity}
\end{equation}
The right-hand side is the $\Sigma^{-1}$-weighted annihilator of the nuisance space spanned by $H$. Concentrating $\delta$ out of the full GLS criterion in \eqref{eq:full-gls} therefore gives the criterion
\[
(\widehat\Pi-G\beta)'
 C'(C\Sigma C')^{-1}C
(\widehat\Pi-G\beta).
\]
Using $CG=D$ and $C\widehat\Pi=\widehat b$, this is exactly
\[
(\widehat b-D\beta)'\Omega_q^{-1}
(\widehat b-D\beta),
\qquad \Omega_q=C\Sigma C'.
\]
The minimizers in $\beta$ coincide. Replacing population covariance matrices by consistent estimates changes the resulting linear maps by $o_p(1)$, and their product with the $O_p(n^{-1/2})$ first-step error changes the estimators by $o_p(n^{-1/2})$.
\end{proof}

\subsection{Asymptotic theory}

\begin{proof}[Proof of Proposition \ref{prop:consistency}]
For each $t$, the sample objective in \eqref{eq:sample-first-stage} is a finite-dimensional convex M-estimation criterion. The check loss is Lipschitz, and Assumption \ref{ass:sampling} supplies an integrable envelope on compact parameter sets, so a uniform law of large numbers applies. Assumption \ref{ass:unique-projection} gives a unique population minimizer. Moreover, $\rho_q(u)\geq\min\{q,1-q\}|u|$, and positive definiteness of $\E[X_iX_i']$ implies
\[
\inf_{\|a\|=1}\E|X_i'a|>0.
\]
The population objective is therefore coercive, so no compact parameter-space restriction is needed. The standard convex argmin theorem yields $\widehat\pi_t(q)\toP\pi_{0t}(q)$ for every $t$. Since $T$ is fixed, stacking and premultiplying by $C$ give \eqref{eq:first-step-consistency}. Finally, \eqref{eq:generic-md-closed}, $b_0=D\beta_0$, $A_n\toP A$, and the full column rank of $D$ imply \eqref{eq:md-consistency} by continuous mapping.
\end{proof}

\begin{proof}[Proof of Theorem \ref{thm:first-stage-asymptotic}]
For each fixed $t$, $\widehat\pi_t(q)$ is an ordinary cross-sectional quantile-regression estimator with independent observations across $i$. The objective may be misspecified as a model of the conditional quantile, but it is correctly centered at the unique population check-loss minimizer $\pi_{0t}(q)$. Under Assumption \ref{ass:qr-regularity}, the standard convexity argument based on Knight's identity yields
\begin{equation}
\sqrt n\{\widehat\pi_t(q)-\pi_{0t}(q)\}
 =J_q^{-1}\frac{1}{\sqrt n}\sum_{i=1}^n
 X_i\psi_q\{U_{it}(q)\}+o_p(1).
\label{eq:single-bahadur}
\end{equation}
The same $J_q$ appears for every $t$ because $U_{it}(q)\mid X_i$ is stationary by Lemma \ref{lem:composite-stationarity} and Proposition \ref{prop:projection-representation}. Since $T$ is fixed, the $T$ representations can be stacked, giving \eqref{eq:bahadur-stacked}. The multivariate central limit theorem applied to the iid cluster scores $g_i(q)$ gives \eqref{eq:SigmaPi}. Premultiplication by the fixed matrix $C$ and use of $C\Pi_0=D\beta_0$ yield \eqref{eq:Omegaq}.
\end{proof}

\begin{proof}[Proof of Corollary \ref{cor:md-asymptotic}]
From \eqref{eq:generic-md-closed} and $b_0=D\beta_0$,
\begin{align*}
\sqrt n\{\widehat\beta_A(q)-\beta_0\}
&=(D'A_nD)^{-1}D'A_n
 \sqrt n\{\widehat b(q)-D\beta_0\}.
\end{align*}
Slutsky's theorem and \eqref{eq:Omegaq} give \eqref{eq:generic-md-limit}--\eqref{eq:generic-md-var}. Substituting $A=\Omega_q^{-1}$ simplifies the variance to \eqref{eq:optimal-md-var}.
\end{proof}

\begin{proof}[Proof of Proposition \ref{prop:covariance-consistency}]
The kernel law of large numbers and continuity of $f_{U(q)\mid X}$ at zero imply that the infeasible version of \eqref{eq:Jhat-period}, formed with $U_{it}(q)$, converges to $J_q$. The generated-residual error is negligible under $\sqrt n h_n^2\to\infty$, the Lipschitz property of $K$, and the $O_p(n^{-1/2})$ rate of $\widehat\pi_t(q)$. Averaging over the fixed number of periods therefore gives $\widehat J_q\toP J_q$.

Because the conditional distribution of $U_{it}(q)$ is continuous at zero, replacing $U_{it}(q)$ by $\widehat U_{it}(q)$ in the indicator score affects a vanishing fraction of observations in probability. The moment condition in Assumption \ref{ass:kernel} and a cluster law of large numbers then yield $\widehat S_q\toP S_q$. Positive definiteness and continuous mapping imply the remaining statements in \eqref{eq:all-cov-consistency} and \eqref{eq:V-consistency}.
\end{proof}

\begin{proof}[Proof of Proposition \ref{prop:latent-covariance}]
Conditional stationarity and finite second moments imply that $\E(V_{it}\mid X_i,A_i)$ is the same function of $(X_i,A_i)$ for every $t$. Hence $\Cov(A_i,V_{it})=c_{AV}$ does not depend on $t$. Expanding the covariance of $R_i=A_i\1_T+V_i$ gives
\[
\Gamma_R
 =\Gamma_V+\{\Var(A_i)+2c_{AV}\}\1_T\1_T'.
\]
This is \eqref{eq:covariance-decomposition} with $\lambda_0=\Var(A_i)+2c_{AV}$. Premultiplying and postmultiplying by any $L$ such that $L\1_T=0$ removes the rank-one component and yields \eqref{eq:difference-covariance}. Consistency of $L\widehat\Gamma_RL'$ follows from consistency of $\widehat\Gamma_R$ and continuous mapping.
\end{proof}

\begin{proof}[Proof of Proposition \ref{prop:J-test}]
Let $P_D^{\Omega}$ denote the $\Omega_q^{-1}$-weighted projection onto the column space of $D$. Standard minimum-distance algebra gives
\[
J_n(q)
 =Z_n'\Omega_q^{-1/2}
 (I-P_D^{\Omega})
 \Omega_q^{-1/2}Z_n+o_p(1),
\]
where $Z_n=\sqrt n\{\widehat b(q)-D\beta_0\}\toD N(0,\Omega_q)$. The idempotent matrix $I-P_D^{\Omega}$ has rank $pr-p=p\{T(T-1)-1\}$. The quadratic form therefore converges to the stated chi-square distribution. Consistency of $\widehat\Omega_q$ permits replacement of population by estimated weights.
\end{proof}

\section{Additional implementation details}
\label{app:implementation}

\subsection{Explicit equal-weight formula}

For each coordinate of $\beta_0$, the identity-weight full-system problem is an ordinary least-squares regression of the $T^2$ first-step coefficients on a regressor-block fixed effect and an indicator for the diagonal element. Concentrating out the block fixed effects gives
\begin{equation}
\widehat\beta_{\mathrm{EW}}(q)
 =\frac{1}{T(T-1)}
 \sum_{s=1}^T\sum_{t\neq s}
 \left\{\widehat\pi_s^{(s)}(q)-\widehat\pi_t^{(s)}(q)\right\},
\label{eq:appendix-equal-weight}
\end{equation}
which is \eqref{eq:equal-weight-estimator}. For $T=2$, this becomes
\begin{equation}
\widehat\beta_{\mathrm{EW}}(q)
 =\frac{1}{2}\left[
 \{\widehat\pi_1^{(1)}(q)-\widehat\pi_2^{(1)}(q)\}
 +\{\widehat\pi_2^{(2)}(q)-\widehat\pi_1^{(2)}(q)\}
 \right].
\label{eq:T2-average}
\end{equation}
Thus an estimator based on only one of the two blocks is consistent but is not generally the identity-weight minimum-distance estimator.

\subsection{Recommended reporting}

For empirical and Monte Carlo work, it is useful to report:
\begin{enumerate}[label=(\roman*)]
\item every individual contrast $\widehat b_{st}(q)$ and its standard error;
\item the equal-weight and efficient pooled estimates;
\item the condition numbers of $\widehat J_q$ and $\widehat\Omega_q$;
\item the analytic and pairs-bootstrap standard errors;
\item the $J$ statistic and its degrees of freedom; and
\item sensitivity to the quantile index $q$, the density bandwidth, and the use of pooled versus period-specific Jacobian estimates.
\end{enumerate}
These diagnostics distinguish weak within-panel design variation, instability in the density estimate, and failure of the stationarity-induced projection restrictions.

\end{document}